\documentclass{article}
\usepackage{mathrsfs}
\usepackage{a4wide}
\usepackage{soul}
\usepackage{enumerate}
\usepackage{amsmath,amsfonts,amssymb,amsthm}
\usepackage{color}

\def \beq{\begin{equation}}
\def \eeq{\end{equation}}

\def\and {{\rm \; and \;}}

\newcommand{\R}{{\mathbb R}}

\newcommand{\y}{{\bf y}}
\newcommand{\C}{\mathbb{C}}
\newcommand{\x}{{\bf k}}
\newcommand{\xy}{{\bf x}}
\newcommand{\yy}{\underline{y}}

\newcommand{\bc}{{\bf c}}
\newcommand{\A}{{\bf A}}

\newcommand{\Z}{{\mathbb Z}}
\newtheorem{theorem}{Theorem}[section]
\newtheorem{definition}[theorem]{Definition}
\newtheorem{proposition}[theorem]{Proposition}

\newtheorem{lemma}[theorem]{Lemma}

\theoremstyle{definition}
\newtheorem{remark}[theorem]{Remark}
\newtheorem{problem}[theorem]{Problem}

\numberwithin{equation}{section}

\begin{document}

\noindent 
\begin{center}
\textbf{\Large On the construction of composite Wannier functions}

\vspace{0.5cm}

March 9, 2016
\end{center}

\vspace{0.5cm}

\noindent 

\begin{center}

\small{
Horia D. Cornean\footnote{Department of Mathematical Sciences,     Aalborg University, Fredrik Bajers Vej 7G, 9220 Aalborg, Denmark}, Ira Herbst\footnote{Department of Mathematics, 
University of Virginia, 
Charlottesville, VA 22903, USA }, Gheorghe Nenciu\footnote{Institute of Mathematics of the Romanian Academy,
P.O. Box 1-764, RO-70700, Bucharest, Romania}}

\end{center}

\vspace{0.5cm}

\noindent

\begin{abstract}

We give a constructive proof for the existence of an $N$-dimensional Bloch basis which is both smooth (real analytic) and periodic with respect to its $d$-dimensional quasi-momenta, when $1\leq d\leq 2$ and $N\geq 1$. The constructed Bloch basis is conjugation symmetric when the underlying projection has this symmetry, hence the corresponding exponentially localized composite Wannier functions are real. In the second part of the paper we show that by adding a weak, globally bounded but not necessarily constant magnetic field, the existence of a localized basis is preserved.      
\end{abstract}

\tableofcontents

\newpage

\section{Introduction and main results}\label{intro}

\subsection{Generalities}

The Wannier functions (as bases of localized functions spanning subspaces corresponding to energy bands in periodic solids) have been playing a central role since their introduction in 1937 in both qualitative and quantitative aspects of the one-electron theory of solid state physics.
In particular, they are the key ingredient in obtaining effective tight-binding Hamiltonians, a topic related to the Peierls-Onsager substitution.  
After the seminal paper of Marzari and Vanderbilt \cite{MV}, they have become a powerful tool in {\em ab initio} computational studies of electronic properties of materials (see \cite{N2}, \cite{MMYSV}, \cite{FMP} and references given there). 

The existence and construction of {\em exponentially localized} Wannier functions (ELWF) has been one of the fundamental problems in solid state physics. Let us remind the reader that  prior to the 
crucial observation made by Thouless \cite{Th} that the very existence of ELWF implies that the quantum Hall current 
of the corresponding band vanishes, the fact that ELWF might not exist was largely overlooked by physicists. 

It turns out that due to several mathematical subtleties,  the proof of exponential localization is difficult and depends on $d$ (the dimension of the configuration space) among other things. Accordingly,  
 the progress in proving the existence of 
ELWF in all cases of physical interest was slow. For example, the existence of ELWF for composite energy bands of time-reversal invariant Hamiltonians in two and three dimensions has been  only recently achieved \cite{P}, \cite{BPCMM}.

The first rigorous result about the existence of ELWF in one dimension was obtained  by Kohn in his 1959  classic paper \cite{K}. 
Using a technique based on ordinary differential equations, Kohn considered simple bands in crystals with a center of inversion and obtained a complete solution: he constructed Wannier functions which were real, symmetric or antisymmetric with respect to an appropriate reflection and with an optimal exponential decay. Kohn's method does not generalize to higher dimensions and the next major step was made  
by des Cloizeaux. In a couple of basic papers \cite{dCl1}, \cite{dCl2}  in 1964 which are still worth reading, he showed 
that the existence of ELWF for composite bands in crystals of arbitrary dimension can be achieved in two steps:

\begin{enumerate}[I]
\item. To each isolated energy band one associates a  family $P({\bf k}), {\bf k}\in {\R}^d$, of rank $N$ 
orthogonal projections ($N=1$ for simple bands and $1<N<\infty$ for composite bands) which are $\Z^d$-periodic and jointly
analytic in a complex tubular neighborhood $I$ of ${\mathbb R}^d$. These projections live in a separable Hilbert space $\mathcal{H}$. 

\item. Show that the vector bundle  $\mathbb{B}:=\{(\x, {\bf v}):\; \x\in \mathbb{T}^d,\; {\bf v}\in {\rm Ran}P({\x})\} $ is {\bf trivial} (= the subspace admits a continuous orthonormal basis). More precisely, one has to show there exist  $N$ vectors 
$\{\Xi_j({\bf k})\}_{j=1}^N\subset \mathcal{H}$ which are continuous,  $\mathbb{Z}^d$-periodic and form an orthonormal basis of 
$ {\rm Ran}(P({\bf k}))$ for all ${\bf k}\in\R^d.$
\end{enumerate}

Step I. is quite general and it is based on the Bloch-Floquet-Gelfand transform and analytic perturbation theory.  It works for all standard Schr\"odinger operators with periodic potentials, including their discrete and pseudo-differential 
variants (see e.g.\cite{dCl1},\cite{N2}, \cite{FMP}, \cite{dN-G} as well as section \ref{subs1.2}  below).

Step II. is a subtle and more demanding mathematical problem, due to the fact that vector bundles can be non-trivial. A well known example is the tangent space to any even-dimensional unit sphere: the existence of a smooth orthonormal basis would contradict the Hairy-Ball Theorem. des Cloizeaux realized that in proving step II the obstructions have topological origin and showed that 
they are not present for simple bands ($N=1$) in the absence of a magnetic field or when the periodic electric potential has a center of inversion.
The restriction to crystals with a center of inversion  was removed in 1983 by Nenciu \cite{N0} who proved  by 
operator analytic methods that at the abstract level, for arbitrary $d$ and $N=1$ the obstructions to the triviality of the vector bundle $\mathbb{B}$ 
are absent if $P(\x)$ satisfies
\begin{equation}\label{cs}
\theta P({\bf k})\theta=P(-{\bf k})
\end{equation}
for some antiunitary involution $\theta$. The proof
is constructive and leads to optimally localized, real Wannier functions
\cite{N2}. A simpler proof was later given by Helffer and Sj\"ostrand \cite{He-Sj}; however, it seems that it does not provide the optimal exponential decay for the Wannier functions. 

For a Schr\"odinger operator with a real symbol (i.e. without a magnetic potential), \eqref{cs} is insured by the fact that it commutes with complex conjugation i.e. it  obeys time reversal symmetry. Concerning the case $N>1$, it has 
been sugested \cite{N2}, \cite{NN2} that exploiting \eqref{cs} in the context of characteristic classes theory in combination with 
some deep results in the theory of analytic functions of several complex variables might lead to the existence of ELWF. Indeed, 
if $ d=2,3$ and if $P({\bf k})$ satisfies \eqref{cs}, the triviality of the bundle was proved by Panati \cite{P}, providing a  non-constructive proof of the existence of ELWF for all cases of physical interest \cite{BPCMM}. { In higher dimensions, examples of non-trivial bundles have been given in \cite{dN-G}. }

 While concerning the non-constructive existence of ELWF the situation is satisfactory, there are still many interesting questions left open in the area. Below we list a few of them:
 
 \begin{enumerate}

\item As it stands in \cite{P},  the triviality of the bundle generated by $P({\bf k})$ for $d=2,3$ and $N>1$ is an abstract existence result, but at least for 
computational aspects, a constructive proof is needed. An important step in this direction has been very recently made in \cite{FMP}: the authors construct real { valued} Wannier functions which decay faster than any polynomial. 

\item All the results outlined above (as well as the very definition of Wannier functions) require strict periodicity of 
the underlying Hamiltonian. However, as stressed by many authors (see e.g. \cite{KO}, \cite{GK}, \cite{Ki} ) results on the construction of 
exponentially localized bases in non-periodic (or at least nearly periodic) systems are highly desirable.

\item For theoretical as well as computational purposes it is highly desirable to construct localized Wannier functions inheriting as
much as possible the symmetries of the underlying Hamiltonian, e.g. time reversal invariance should lead to real valued  Wannier functions.

\item At least from a computational point of view it is important to construct Wannier functions with best localization properties. In this context, one has to define precise localization criteria and to discuss the relations between them \cite{MV, CNN, MP, PP}.

\item In dimension one, Wannier functions can be constructed as the eigenvectors of the position operator restricted to the range of the Riesz projection of the isolated band \cite{NN1}. Can this method be extended to higher dimensions? { Some recent progress on the so-called `radial localization problem' has been made by Prodan \cite{Pr}. }
\end{enumerate}

The aim of our paper is to add some results concerning the first three questions above. Our first main result is a construction (assuming \eqref{cs})
 of real valued ELWF for $d=2$ and $N>1$. Our method is quite different from the one in  \cite{FMP}; in particular, it does not require a preliminary reduction to $N=2$, and builds on the 
operator analytic approach in \cite{N0}, \cite{N2}. 
{ Note that the results of our paper are not  applicable to the so-called `fermionic' case related to $\Z_2$ topological insulators \cite{SV, FMP2}. However, our methods can be adapted in order to deal with the two dimensional fermionic case (\cite{CMT}, in preparation). 

}

The
discussion above is about the case when no magnetic fields are present. { While the existence of Wannier bases for real and periodic Schr\"odinger operators can be reduced to the study of the existence of smooth and periodic Bloch bases, for non-periodic systems the situation is more complicated. Nevertheless, for a large class of non-periodic perturbations of periodic systems it can be proved that localized bases still exist, see \cite{NN2}; the main idea behind it is based on a `continuity' argument which interpolates between the periodic and non-periodic operators.}

Our second main result gives the construction of bases of exponentially 
localized functions for the case when a weak, globally bounded magnetic field is added, taking for granted the 
existence of such bases for the zero magnetic field case.

The content of the paper is as follows. In paragraph \ref{subs1.2} we outline how one arrives at the family $P({\bf k})$ starting from a periodic tight-binding Hamiltonian/ Schr\"odinger operator. In paragraph \ref{subs1.3} we formulate the main results. Sections \ref{sec2} and \ref{sec3} are devoted to proofs. 

\subsection{From time reversal symmetric Hamiltonians to periodic $P({\bf k})$ satisfying conjugation symmetry}\label{subs1.2}

Let $\mathcal{B}\subset \R^d$ be a finite set containing $D$ points and let  
$$\Lambda:=\mathcal{B}+\mathbb{Z}^d$$
be the discrete configuration space. We assume that every point $\y\in \Lambda$ can be uniquely written as:
$$\y=\yy +\gamma,\quad \yy\in \mathcal{B},\quad \gamma\in \mathbb{Z}^d.$$ 
Consider the Hilbert space $l^2(\Lambda)$ and let $H_0$ be a bounded self-adjoint operator  
acting on this space, which commutes with the translations acting on $\mathbb{Z}^d$. It is uniquely determined by the 'matrix elements' 
$$H_0(\yy,\gamma;\yy',\gamma')= H_0(\yy,\gamma-\gamma';\yy',0),\quad \forall \yy,\yy'\in \mathcal{B},\quad \forall \gamma,\gamma'\in \mathbb{Z}^d.$$
Consider the unitary Bloch-Floquet-Gelfand transform 
$$l^2(\Lambda)\ni \Psi \mapsto (G\Psi)(\yy;\x):=\sum_{\gamma\in \mathbb{Z}^d}e^{2\pi i \x\cdot \gamma}\Psi(\yy+\gamma)\in L^2([-1/2,1/2]^d;\C^D).$$
We have 
$$GH_0G^*=\int_{[-1/2,1/2]^d}^\oplus h(\x) d\x,\quad h(\x)\in \mathcal{L}(\C^D),\quad h(\yy,\yy';\x):=\sum _{\gamma\in \mathbb{Z}^d}e^{2\pi i \x\cdot \gamma}H_0(\yy,\gamma;\yy',0).$$
The $D\times D$ matrix $h(\x)$ will have $D$ eigenvalues (Bloch energies), whose range determine the spectrum of $H_0$. We see that $h(\x)$ can be periodically extended to the whole of $\R^d$ and that the decay of $H_0( \yy,\gamma;\yy',0)$ as function of $\gamma$ tells us how regular $h(\x)$ is. For example, exponential localization of $H_0$ gives real analyticity for $h$. Moreover, if $H_0$ is a real kernel, then we immediately obtain:
$$h(\yy',\yy;\x)=\overline{h(\yy,\yy';\x)}=h(\yy,\yy';-\x),\quad \overline{h(\x)f}=h(-\x)\overline{f},$$
a property which from now on will be called conjugation symmetry. 
We will see in Lemma \ref{lema-dec-1} that for a self-adjoint matrix, conjugation symmetry is equivalent to:
$${}^t h(\x)=h(-\x),\quad \x\in \R^d,$$
where ${}^t h(\x)$ means transposition.

Assume that the union of the ranges of $N<D$ eigenvalues $\lambda_j(\x)$ of $h(\cdot)$ is separated from the range of the other eigenvalues, i.e. there exists a gap in the spectrum of $H_0$. We can define $P(\x)$ to be the Riesz projection associated to these $N$ eigenvalues. $P(\cdot)$ is periodic and has the same type of smoothness as $h$.  Moreover, 
because 
$${}^t\{(h(\x)-z)^{-1}\}=(h(-\x)-z)^{-1},$$ 
it follows that $P(\cdot)$ has the same conjugation property. 

Consider an orthonormal system of $N$ vectors  $\Xi_j(\x)\in\C^D$ which form a basis for ${\rm Ran}(P(\x))$, { continuous in $\x\in [-1/2,1/2]^d$}. Then we can define:
\begin{equation}\label{zumba0}
w_j(\y)=w_j(\yy+\gamma):=\int_{[-1/2,1/2]^d} e^{-2\pi i \x\cdot \gamma}\Xi_j(\yy;\x)d\x,\quad 1\leq j\leq N,\quad \yy\in \mathcal{B},\quad \gamma\in \mathbb{Z}^d.
\end{equation}
It is known that the set obtained by translating them as follows: 
$$\{w_j(\cdot -\gamma):\; \gamma\in \mathbb{Z}^d,\; 1\leq j\leq N\}\subset l^2(\Lambda)$$
forms an orthonormal basis for the range of the spectral projection of $H_0$ corresponding to the part of the spectrum given by
$$\sigma_0:=\bigcup_{j=1}^N {\rm Ran}\;\lambda_j{(\cdot)}.$$

A similar problem can be formulated for continuous operators. Let $V$ be a $\Z^d$-periodic, real and bounded potential. The operator $H=-\Delta +V$ has purely absolutely continuous spectrum. 

Let $\Omega:=[-1/2,1/2]^d$. Consider the Bloch transform 
$$L^2(\R^d)\ni \Psi \mapsto (F\Psi)(\yy;\x):=\sum_{\gamma\in \mathbb{Z}^d}e^{2\pi i \x\cdot\gamma}\;\Psi(\yy+\gamma)\in L^2(\Omega;L^2(\Omega)).$$
We have (in the formula below, $\x$-BC is a short-cut for $\x$-dependent boundary conditions $f(\yy+\gamma)=e^{-2\pi i\x\cdot \gamma} f(\yy)$, when both $\yy$ and $\yy+\gamma$ belong to the boundary of $\Omega$):
$$FHF^*=\int_{\Omega}^\oplus h_\x d\x,\quad h_\x=-\Delta+V \; {\rm with \; \x{\rm -BC}\; in } \; L^2(\Omega).$$
Each fiber Hamiltonian $h_\x$ has purely discrete spectrum. Assume that the range of $N$ eigenvalues form an isolated spectral island $\sigma_0$. { Using Combes-Thomas exponential estimates and elliptic regularity (see Proposition 3.1 in \cite{CN2})}, one can show that {  the Riesz projection $\Pi$ corresponding to $\sigma_0$ } has a real integral kernel $\Pi(\xy,\xy')$ which is jointly continuous and exponentially localized around the diagonal, i.e. there exists $\alpha>0$ such that $|\Pi(\xy,\xy')|\leq e^{-\alpha |\xy-\xy'|}$. Then 
$$F\Pi F^*=\int_{\Omega}^\oplus P_\x d\x,\quad P_\x(\yy,\yy')=\sum _{\gamma\in \mathbb{Z}^d}e^{2\pi i \x\cdot \gamma}\Pi(\yy+\gamma;\yy') \; {\rm in}\; L^2(\Omega).$$
Again, we see that $P_\x$ has the conjugation symmetry, is real analytic and $\Z^d$-periodic in $\x$. Note that \eqref{zumba0} also makes sense in the continuous case, with the only difference that $\yy\in \Omega$. 

{ In the continuous case one can also consider the Bloch-Zak transform 
$$L^2(\R^d)\ni \Psi \mapsto (F_Z\Psi)(\yy;\x):=\sum_{\gamma\in \mathbb{Z}^d}e^{2\pi i \x\cdot(\yy+\gamma)}\;\Psi(\yy+\gamma)\in L^2(\Omega;L^2(\Omega)).$$
At fixed $\x$ and for $\Psi\in C_0^\infty(\R^d)$, the function 
$(F_Z\Psi)(\cdot;\x)$ is $C^\infty(\overline{\Omega})$ and periodic. We have:
$$F_ZHF_Z^*=\int_{\Omega}^\oplus \tilde{h}_\x d\x,\quad \tilde{h}_\x=(-i\nabla +2\pi\x)^2+V \; {\rm with \; {\rm periodic-BC}\; in } \; L^2(\Omega).$$
The Riesz projection $\Pi$ is decomposed as:
$$F_Z\Pi F_Z^*=\int_{\Omega}^\oplus \Pi_\x d\x,\quad \Pi_\x(\yy,\yy')=\sum _{\gamma\in \mathbb{Z}^d}e^{2\pi i \x\cdot (\yy+\gamma -\yy')}\Pi(\yy+\gamma;\yy') \; {\rm in}\; L^2(\Omega).$$
We see that the range of $\Pi_\x$ consists of $\Omega$-periodic functions in $\yy$, but $\Pi_\x$ is no longer periodic in $\x$. Instead, denoting by $\tau_\lambda$ the unitary operator acting on $L^2(\Omega)$  given by $(\tau_\lambda f)(\yy)=e^{2\pi i\lambda \cdot \yy}f(\yy)$ we have the identity: 
$$\Pi_{\x+\lambda}=\tau_\lambda \Pi_\x \tau_\lambda^{-1},\quad \forall \lambda\in \Z^d.$$
}

If the $\Xi_j$'s are eigenvectors of $h(\x)$, then the corresponding $w_j$'s (see \eqref{zumba0}) are called Wannier functions. In general, the Wannier functions are just square integrable, but in many physical applications it is important to construct better spatially localized Wannier functions. For example, if we can simultaneously choose the $\Xi_j$'s to be periodic and real-analytic, then a standard Paley-Wiener argument shows that the Wannier functions are exponentially localized. Moreover, if we also have  $\overline{\Xi_j(\x)}=\Xi_j(-\x)$, then the Wannier functions are real.  

If the $\Xi_j$'s are not eigenvectors for the fibre Hamiltonian but they form an orthonormal basis for ${\rm Ran}(P(\x))$, we can still define the corresponding $w_j$'s. They are called composite Wannier functions.

\subsection{Construction of composite Wannier functions}\label{subs1.3}

If $\R^d\ni \x\mapsto P(\x)\in \mathcal{L}(\mathcal{H})$ is a  projection valued continuous map in
some separable complex Hilbert space $\mathcal{H}$, then the dimension of ${\rm Ran}(P(\x))$ is constant. We list four basic problems. 

\begin{problem}\label{problema1}
{\it Let $d\geq 1$ and let $\R^d\ni \x\mapsto P(\x)\in \mathcal{L}(\mathcal{H})$ be an orthogonal projection valued map which is $\mathbb{Z}^d$-periodic and belongs to $C^n(\R^d)$, $n\geq 1$. Let ${\rm dim}{\rm Ran}(P(\x))=N<\infty$. Can one find a system of $N$ vectors 
$\{\Xi_j(\x)\}_{j=1}^N$ which form an orthonormal basis in ${\rm
  Ran}(P(\x))$ and which are both in $C^n(\R^d)$ and
$\mathbb{Z}^d$-periodic? }
\end{problem}

\vspace{0.2cm}

\begin{problem}\label{problema2}
{\it If $d\geq 1$, let $\R^d\ni \x\mapsto P(\x)\in \mathcal{L}(\mathcal{H})$ be an orthogonal projection valued map, $\mathbb{Z}^d$-periodic and jointly analytic in a tubular complex neighborhood $I$ of $\R^d$. Let ${\rm dim}{\rm Ran}(P(\x))=N<\infty$.
Can one find a system of $N$ vectors 
$\{\Xi_j(\x)\}_{j=1}^N$ which are jointly analytic in a (possibly smaller than $I$)  tubular complex neighborhood $I'$ of $\R^d$, 
are $\mathbb{Z}^d$-periodic and form an orthonormal basis of ${\rm
  Ran}(P(\x))$ when $\x\in\R^d$? }
\end{problem}

\vspace{0.2cm}

Most of the results of this paper will only be valid if some conjugation symmetries hold. Here is a list of definitions: 

\begin{definition}\label{CS}
We say that a family $H(\cdot)$ { of linear operators} acting on a complex Hilbert space $\mathcal{H}$ has the $CS$ property if there exists an anti-unitary operator $K$ such that 
\begin{equation}\label{septemb1}
K^2={\rm Id},\quad KH(\x)K=H(-\x),\quad \x\in \R^d.
\end{equation}
\end{definition}

The most common example of such a $K$ in $l^2(\mathbb{Z})$ is the complex conjugation $Kf=\overline{f}$:
\begin{equation}\label{jan0}
\overline{H(\x)f}=H(-\x)\overline{f},\quad \x\in\R^d.
\end{equation}
In fact, up to a change of basis, every such $K$ is given by a complex conjugation. 

\begin{definition}\label{CS'}
We say that a family of invertible operators $U(\cdot)$ acting on a complex Hilbert space $\mathcal{H}$ has the $CS '$ property if there exists an anti-unitary operator $K$ such that 
\begin{equation}\label{septemb1''}
K^2={\rm Id},\quad KU(\x)K=[U(-\x)]^{-1},\quad \x\in \R^d.
\end{equation}
\end{definition}

The third problem is the following:
\begin{problem}\label{problema3}
{\it Assume that $P(\cdot)$ has the $CS$ property. Can we construct a basis as in the previous two problems,  which also obeys the property $K\Xi_j(\x)=\Xi_j(-\x)$? }
\end{problem}

\vspace{0.5cm}

Finally, let us formulate a more general problem as in \cite{P, FMP}, { motivated by the Bloch-Zak transform from the continuous case.} Assume that we have some separable complex Hilbert space $\mathcal{H}$ with a conjugation $K$ as in \eqref{septemb1}, and a family of unitary operators $\tau_\lambda$ with $\lambda \in \mathbb{Z}^d$ such that 
\begin{equation}\label{septemb2}
\tau_0={\rm Id},\quad \tau_{\lambda}\tau_\mu =\tau_{\lambda+\mu}(=\tau_\mu\tau_{\lambda}),\quad \tau_\lambda^*=\tau_{\lambda}^{-1}(=\tau_{-\lambda}),\quad K\tau_\lambda =\tau_{-\lambda}K.
\end{equation}

\begin{problem}\label{problema4}
{\it If $d\geq 1$, let $\R^d\ni \x\mapsto \Pi(\x)$ be an orthogonal projection valued map in $\mathcal{H}$ such that:
$$ \Pi(\x+\lambda)=\tau_{\lambda}\Pi(\x)\tau_{\lambda}^*,\quad K\Pi(\x)K=\Pi(-\x),\quad \forall \x\in\R^d,\quad \forall \lambda\in \mathbb{Z}^d. $$
Assume that the map is jointly analytic in a tubular complex neighborhood $I$ of $\R^d$. 
Assume that the dimension of ${\rm Ran}(\Pi(\x))$ equals
$ N<\infty$. Can one find a system of $N$ vectors 
$\{\Psi_j(\x)\}_{j=1}^N$ which are jointly analytic in a (possibly smaller than $I$)  tubular complex neighborhood $I'$ of $\R^d$, which form an orthonormal basis of ${\rm
  Ran}(\Pi(\x))$ { when $\x\in\R^d$}, obey $\tau_\lambda \Psi_j(\x)=\Psi_j(\x+\lambda)$ and $K \Psi_j(\x)=\Psi_j(-\x)$ for all $\x\in\R^d$ and $\lambda\in \mathbb{Z}^d$?}
\end{problem}

{ 
\begin{remark}\label{remarcazak}
We will show in Subsection \eqref{reducere4} that this more general problem can be reduced to the periodic case. However, other more subtle properties like the existence of a canonical Berry connection only appear in the Zak picture, see \cite{FCG} for details.  
\end{remark}
}

\vspace{0.2cm}

Here is our first main result.
\begin{theorem}\label{metateorema} The following statements hold true: 

\noindent {\rm (i)} Problem \ref{problema4} can be reduced to the first three; see Subsection \ref{reducere4};

\noindent {\rm (ii)} All problems can be reduced to finite dimensional Hilbert spaces; see Subsection \ref{finidimi};

\noindent {\rm (iii)} Assume that the family $P(\x)$ is real analytic. If we can construct a {\bf continuous} solution to our problems, then it can be made {\bf real analytic}, preserving all its other properties; see Lemmas \ref{lema3} and \ref{lema3'};  

\noindent {\rm (iv)} If $d=1$, a solution exists even in the absence of the $CS$ property; see \eqref{solprob1} and \eqref{matrix2}; 

\noindent {\rm (v)} If $d=2$, a solution can be constructed if the $CS$ property holds true. The main technical result is Proposition \ref{teorema-u}.
\end{theorem}

\vspace{0.2cm}

\begin{remark}\label{remarcamartie}
Let us try to explain very roughly the main ideas behind the proof of { (v)}. The construction is inductive in $d$; using { (iv)} we can construct some basis vectors $\Psi_j(k_1,k_2)$ which are continuous in $\x$ but periodic only in $k_2$. Due to the periodicity of the projection $P(\x)$, the basis at $k_1=1/2$ is related to the basis at $k_1=-1/2$ by a matching $N\times N$ unitary matrix $\beta(k_2)$ which is  continuous and periodic in $k_2$, and whose matrix elements have the symmetry $[\beta(k_2)]_{mn}=[\beta(-k_2)]_{nm}$. The main question is how to rotate each $\Psi_j(k_1,k_2)$ using $\beta(k_2)$ such that the new vectors to be equal at $k_1=\pm 1/2$. This is done in Proposition \ref{teorema-u}, a result interesting in itself. Also, a related question is finding sufficient conditions for a matrix like $\beta(k_2)$ to admit a  periodic {\bf and } continuous logarithm in $k_2$. The main obstacle to the existence of such a logarithm is the possible crossing of eigenvalues of $\beta(k_2)$ when $k_2$ varies, see \eqref{nologaritm} for a generic example. One of the crucial ingredients in order to circumvent the crossing problem is the Analytic Rellich Theorem \cite{R-S4} which allows us to continuously follow the eigenvalues through crossings, also preserving periodicity in $k_2$. 
\end{remark}

\begin{remark}\label{remarcamartie2}
The lack of a Rellich Theorem in more than one dimension is also the main reason for which our proof cannot be generalized to $d=3$. We conjecture though that a variant of Proposition \ref{teorema-u} remains true if $d=3$, but the proof becomes much more technical and seems to be necessary to appeal to `avoided crossing' techniques involving Sard's lemma. This construction will be done elsewhere.
\end{remark}

\subsection{The non-zero magnetic field case}\label{subs1.4}

Our second main result of this paper concerns the question as to whether the existence of Wannier-type localized bases is stable under the application of a (weak) non-decaying magnetic field. For simplicity, we only work in $\R^2$. The setting is as follows.  We assume that there exists a background magnetic field, orthogonal to the plane, whose only non-zero component is denoted with $B_0(\xy)$ and obeys:
\begin{align}\label{M-4}
||B_0||_{C^1(\R^2)}<\infty.
\end{align} 
We can always construct its associate transverse gauge: 
\begin{align}\label{M-3}
\A_0(\xy):=\int_0^1 sB_0(s\xy)ds\;{ {\bf x}^{\perp},\quad {\bf x}^\perp :=(x_2,-x_1).}
\end{align} 
Let $V$ be a bounded, real scalar potential and consider the Hamiltonian
\begin{align}\label{M-1}
H_0=(-i\nabla -\mathbf{A}_0)^2+V.
\end{align}  
Assume that the Hamiltonian $H_0$ has an isolated spectral island $\sigma_0$ and let $P_0$ be its corresponding spectral projection. We also assume that the subspace ${\rm Ran}(P_0)$ has an {(exponentially localized) generalized Wannier basis}; let us explain in detail the meaning of this. Let $\Gamma\subset \R^2$ be an infinite set of points such that 
$$\inf\{||\gamma-\gamma'||:\; \gamma,\gamma' \in\Gamma,\; \gamma\neq\gamma'\} >0.$$ 
The above condition implies that $\Gamma$ is countable; in particular, $\Gamma$ can be a periodic lattice. Let $N<\infty$ and assume that the set of functions 
$$w_{j,\gamma}\in {\rm Ran}(P_0),\quad 1\leq j\leq N,\; \gamma\in\Gamma $$
 form an orthonormal basis of ${\rm Ran}(P_0)$. We also assume that there exist $\alpha>0$ and $M<\infty$ such that  
\begin{align}\label{M-2}
\sup_{j,\gamma}\int_{\R^2} |w_{j,\gamma}(\xy)|^2 e^{2\alpha ||\xy-\gamma||}d\xy <M.
\end{align}
Now let us consider a magnetic field perturbation given by $B(\xy)$ such that 
\begin{align}\label{M-5}
||B||_{C^1(\R^2)}\leq 1,\quad \A(\xy):=\int_0^1 sB(s\xy)ds\; {(x_2,-x_1)},
\end{align}
and introduce the perturbed Hamiltonian:
\begin{align}\label{M-7}
H_b:=(-i\nabla -\A_0 -b\A)^2+V,\quad b\in \R.
\end{align}
Below we list a few known facts about the resolvent of $H_b$, see for details \cite{N3, C, CN}:
\begin{itemize}
\item the resolvent $(H_b-z)^{-1}$ has a Schwartz kernel denoted by $(H_b-z)^{-1}(\xy,\xy')$ which is jointly continuous outside the diagonal, has a singularity of the type 
$-\ln (||\xy-\xy'||)$ near the diagonal, and there exist some $C(z)<\infty$ and  $\alpha(z)>0$ such that $|(H_b-z)^{-1}(\xy,\xy')|\leq C(z) e^{-\alpha(z) ||\xy-\xy'||}$ if $||\xy-\xy'||\geq 1$. 
\item Fix a compact $K\subset \rho(H_0)$. Then there exist $b_0>0$, $\alpha >0$ and $C<\infty$ such that for every $0\leq |b|\leq b_0$ we have $K\subset \rho(H_b)$ and uniformly in $\xy\neq \xy'$:
\begin{equation}\label{M-8}
\sup_{z\in K}\left |(H_b-z)^{-1}(\xy,\xy')-e^{ib\phi(\xy,\xy')}(H_0-z)^{-1}(\xy,\xy')\right |\leq C\;|b| \;e^{-\alpha ||\xy-\xy'||},
\end{equation}
where 
\begin{equation}\label{M-9}
\phi(\xy,\xy'):=\int_0^1 \A(\xy'+s(\xy-\xy'))\cdot (\xy-\xy')ds=-\phi(\xy',\xy).
\end{equation}
\end{itemize}
In particular, if $|b|$ is small enough then $H_b$ has an isolated spectral island $\sigma_b$ close to $\sigma_0$, see also \cite{N1} for a different proof of this fact. Denote by $P_b$ the corresponding projection. Here is the second main result of our paper:
\begin{theorem}\label{teorema2} ${}$ 

\noindent {\rm (i)}. If $|b|$ is sufficiently small, then there exist a constant $C<\infty $, a family of continuous functions $\Xi_{j,\gamma,b}$ and $\alpha>0$ satisfying 
\begin{equation}\label{M-10}
|\Xi_{j,\gamma,b}(\xy)|\leq C e^{-\alpha||\xy-\gamma||}
\end{equation}
such that $\{\Xi_{j,\gamma,b}\}_{j\in \{1,...,N\},\gamma\in \Gamma}$ is an orthonormal basis of ${\rm Ran}(P_b)$ and 
\begin{equation}\label{M-11}
\sup_{\gamma,j}||\Xi_{j,\gamma,b}-e^{ib\phi(\cdot,\gamma)}w_{j,\gamma}||\leq C \; |b|.
\end{equation}

\noindent {\rm (ii)}. Moreover, if $\Gamma$ is a periodic lattice, $B_0=0$, $V$ is $\Gamma$-periodic, $B$ is constant and 
 \begin{equation}\label{M-12}
w_{j,\gamma}(\xy)=w_{j,0}(\xy-\gamma)=\overline{w_{j,\gamma}(\xy)},
\end{equation}
then
\begin{equation}\label{M-13}
\Xi_{j,\gamma,b}(\xy)=e^{ib\phi(\xy,\gamma)}\Xi_{j,0,b}(\xy-\gamma)
\end{equation}
and 
\begin{equation}\label{M-14}
\overline{\Xi_{j,0,b}(\xy)}=\Xi_{j,0,-b}(\xy).
\end{equation}
\end{theorem}

\vspace{0.5cm}

\begin{remark}\label{remarca-i}
When $N=1$ (the case of a simple band), Theorem \ref{teorema2}(ii) was 
proved in \cite{N2}. In the current manuscript we use a strategy which is related to the one of \cite{N2}; what we add new is a significant generalization of the results and a substantial simplification of the proof due to the regularized magnetic perturbation theory as developed in \cite{C, N3}. In fact, the estimate \eqref{M-8} is the main new ingredient of the proof of  Theorem \ref{teorema2}. 

Theorem  \ref{teorema2} is interesting in itself but it also plays an important role in the so-called Peierls-Onsager substitution at small magnetic fields \cite{Pe, Lu, PST, FT}. These aspects will be treated elsewhere.    
\end{remark}

\section{Proof of Theorem \ref{metateorema}}\label{sec2}

\subsection{Reduction of Problem \ref{problema4} to the first three}\label{reducere4}

We first show that there exists a family of unitary operators $u_\x$ with { $\x\in\mathbb{C}^d$} such that $u_\lambda=\tau_\lambda$ if $\lambda\in 
\mathbb{Z}^d$, { the map 
$\mathbb{C}^d\ni \x \mapsto  u_{\x} \in \mathcal{L}(\mathcal{H})$ is entire,}
and:
\begin{equation}\label{septemb3}
u_0={\rm Id}, \; u_{\x}u_{\x'} =u_{\x+\x'}(=u_{\x'} u_{\x}),\; u_{\x}^*=u_{\x}^{-1}(=u_{-\x}),\; K u_{\x} =u_{-\x}K=u_{\x}^{-1}K,\quad \forall \x,\x'\in\R^d.
\end{equation}
Denote by $f_j$ the vectors of the standard basis in $\R^d$. The operator $\tau_{f_j}$ is unitary hence the spectral theorem implies that it can be written as 
\begin{equation}\label{septemb4}
\tau_{f_j}=\int_{(-\pi,\pi]} e^{i\phi}dE_j(\phi).
\end{equation}
The operator $M_j= \int_{(-\pi,\pi]} \phi\; dE_j(\phi)$ is bounded, self-adjoint, and $\tau_{f_j}=e^{iM_j}$. 

If $F$ is any $2\pi$-periodic smooth function we define  $\widetilde{F}(e^{it})=F(t)$ and we have
\begin{equation}\label{septemb5}
\widetilde{F}(\tau_{f_j})=\sum_{m\in \mathbb{Z}}\hat{F}(m) (\tau_{f_j})^m,\quad \hat{F}(m):=\frac{1}{2\pi}\int_{-\pi}^\pi F(t)e^{-it m}dt.
\end{equation}
We have $K\tau_{f_j}^mK=(\tau_{f_j})^{-m}$. If $F$ is real, then $\overline{\hat{F}(m)}=\hat{F}(-m)$. This leads to $K\widetilde{F}(\tau_{f_j})K=\widetilde{F}(\tau_{f_j})$ for all $2\pi$-periodic smooth and real functions $F$. By a limiting argument we conclude that the same remains true for the spectral measure of $\tau_{f_j}$, hence $KM_jK=M_j$. Moreover, because the set $\{\tau_{f_j}\}_{j=1}^d$ consists of commuting operators, the same is true for their spectral measures and for $\{M_j\}_{j=1}^d$. Now if $\x=[k_1,...,k_d]\in\R^d$ we define 
$$u_\x:=e^{i (k_1 M_1+...+k_d M_d)}.$$
One can verify that \eqref{septemb3} is obeyed and {because the generators $M_j$ are bounded, the map is also entire}. Now if we define 
$P(\x):=u_{\x}^{-1}\Pi(\x)u_\x$ we see that $P(\cdot)$ is $\mathbb{Z}^d$-periodic, has the same smoothness properties as $\Pi(\cdot)$ and $KP(\x)K=P(-\x)$. Assuming that we have a solution $\{\Xi_j(\x)\}_{j=1}^N$ for Problem \ref{problema3}, then we can easily check that $\Psi_j(\x):=u_\x \Xi_j(\x)$ is a solution for Problem \ref{problema4}.

\subsection{Reduction to finite dimensional Hilbert spaces}
\label{finidimi}

\begin{lemma}\label{lema0}
Fix $d$ and assume that Problem \ref{problema1}  
can be solved if the complex Hilbert space 
has a finite dimension. Then it can also be solved in any 
infinitely dimensional separable complex Hilbert space.  
\end{lemma}

\begin{proof}
Let $L\geq 1$ be an integer and consider a net of $(2L+1)^d$ points 
$$\mathcal{N}_L:=\{\x_n\in [-1/2,1/2]^d:\; \x_n:=[n_1/(2L),\dots n_d/(2L)],\; -L\leq n_j\leq L, \; n_j\in\Z, \;1\leq j\leq d\}.$$

Since $P(\cdot)$ is continuous and $\mathbb{Z}^d$ periodic, then for every $\epsilon>0$ there exists $L_\epsilon$ 
large enough such that for every $\x\in \R^d$ there exists some 
$\x_n\in \mathcal{N}_{L_\epsilon}$ such that $||P(\x)-P(\x_n)||\leq
\epsilon$. Now define the finite dimensional subspace (to be
understood as a sum of linear subspaces)
$$H_\epsilon:=\sum_{\x_n\in \mathcal{N}_{L_\epsilon}}{\rm Ran}(P(\x_n)),$$
and denote by $\Pi_\epsilon$ its corresponding orthogonal projection. Clearly, $\Pi_\epsilon P(\x_n)=P(\x_n)\Pi_\epsilon=P(\x_n)$ 
for all $\x_n$. Define 
$P_\epsilon(\x):=\Pi_\epsilon P(\x) \Pi_\epsilon$. For every $\x\in \R^d$
we have (the choice of $\x_n$ below is such that $||P(\x)-P(\x_n)||\leq
\epsilon$):
$$P_\epsilon(\x)-P(\x)=\Pi_\epsilon \{P(\x)-P(\x_n)\} \Pi_\epsilon+P(\x_n)-P(\x)$$
hence $||P_\epsilon(\x)-P(\x)||\leq 2\epsilon$.  It means that 
$P_\epsilon(\x)$ is close to a projection. If
$\epsilon$ is small enough, then the operator: 

$$\tilde{P}_\epsilon(\x):=-\frac{1}{2\pi i}\int_{|z-1|=1/2}(P_\epsilon(\x)-z)^{-1}dz$$ 
is an orthogonal projection whose range is contained in $H_\epsilon$ (to see this use  $(P_\epsilon(\x)-z)^{-1}+z^{-1}=z^{-1}P_\epsilon(\x)(P_\epsilon(\x)-z)^{-1}$ in the above formula). Using the resolvent formula
$$\tilde{P}_\epsilon(\x)-P(\x)=\frac{1}{2\pi i}\int_{|z-1|=1/2}(P_\epsilon(\x)-z)^{-1}(P_\epsilon(\x)-P(\x))(P(\x)-z)^{-1}dz$$
we obtain $||\tilde{P}_\epsilon(\x)-P(\x)||\leq {\rm const}\;\epsilon$. The dimension of $P(\x)$ is constant and equal 
to $N$, thus the dimension of $\tilde{P}_\epsilon(\x)$ also equals $N$ for
all $\x$, provided $\epsilon$ is small enough. This is true because 
we can construct the Sz.-Nagy unitary which intertwines between $\tilde{P}_\epsilon(\x)$ and $P(\x)$:
\begin{equation}\label{dech11}
U_\epsilon(\x):=\{P(\x)\tilde{P}_\epsilon(\x)+({\rm Id}-P(\x))({\rm Id}-\tilde{P}_\epsilon(\x))\}
\{{\rm Id}-(\tilde{P}_\epsilon(\x)-P(\x))^2\}^{-1/2},
\end{equation}
with $P(\x)U_\epsilon(\x)=U_\epsilon(\x)\tilde{P}_\epsilon(\x)$.

To summarize: we constructed a projection valued map
$\tilde{P}_\epsilon(\cdot)$ with dimension $N$ which enjoys the same
smoothness and periodicity properties as $P(\cdot)$, 
and lives in the finite dimensional space $H_\epsilon$. Our
assumption was that we can solve Problem \ref{problema1} in finite dimensional
Hilbert spaces. 
Hence let $\{\tilde{\Psi}_{j,\epsilon}(\x)\}_{j=1}^N$ be an orthonormal basis of 
${\rm Ran}(\tilde{P}_\epsilon(\x))$ consisting of smooth and periodic
vectors. Then a solution to our problem in the infinite dimensional
Hilbert space will be given by the vectors:
$$\Psi_j(\x):=U_\epsilon(\x)\tilde{\Psi}_{j,\epsilon}(\x).$$\end{proof}

\vspace{0.5cm}

Lemma \ref{lema0} showed how to reduce the problem to a finite dimensional 
Hilbert space. We will now show that the conjugation symmetry is also preserved by the previous construction. 

\begin{lemma}\label{lema00} We refer to the quantities defined in Lemma \ref{lema0}. 
Assume that $P(\cdot)$ has the $CS$ property. Then both  $\tilde{P}_\epsilon(\cdot)$ and $U_\epsilon(\cdot)$ have the same property. Moreover, if 
$K\tilde{\Psi}_{j,\epsilon}(\x)=\tilde{\Psi}_{j,\epsilon}(-\x)$, the same property holds for $\Psi_j(\x)$. 
\end{lemma}

\begin{proof}  First we prove that the subspace $H_\epsilon$ is invariant with respect to the action of $K$, i.e. $\Pi_\epsilon$ commutes with $K$. We constructed the grid $\mathcal{N}_{L_\epsilon}$ so that both $\x_n$ and $-\x_n$ belong to it for all $n$.  { Due to the $CS$ property we have $P(-\x_n)K=KP(\x_n)$, which shows that $K{\rm Ran}(P(\x_n))\subset {\rm Ran}(P(-\x_n))$, hence $K(H_\epsilon)\subset H_\epsilon$. Using $K^2={\rm Id}$ we also obtain $H_\epsilon\subset K(H_\epsilon)$, thus $K(H_\epsilon) = H_\epsilon$. 

Now let $f\in H_\epsilon$ and $g\in H_\epsilon^\perp$ be arbitrary. Then using the anti-unitarity of $K$ together with $Kf\in H_\epsilon$, we obtain:
$$\langle K g| f\rangle=\langle Kf|g\rangle =0.$$
This shows that $K(H_\epsilon^\perp)\subset H_\epsilon^\perp$ and in fact $K(H_\epsilon^\perp)= H_\epsilon^\perp$. Then:
$$\Pi_\epsilon K ({\rm Id} -\Pi_\epsilon)=({\rm Id} -\Pi_\epsilon)K\Pi_\epsilon =0.$$
This implies that $\Pi_\epsilon K=K\Pi_\epsilon$. It follows that $P_\epsilon(\x)=\Pi_\epsilon P(\x) \Pi_\epsilon$ has the $CS$ property. We see that 
$$(P_\epsilon(\x)-z)^{-1}K\Psi =K(P_\epsilon(-\x)-\overline{z})^{-1} \Psi$$
which after integration leads to the conclusion that $\tilde{P}_\epsilon(\x)$ has the $CS$ symmetry. 

The last thing we need to prove is that the Sz.-Nagy unitary $U_\epsilon(\x)$ also has $CS$ symmetry. This is immediately implied by $KP(\x)K=P(-\x)$ and $K\tilde{P}_\epsilon(\x)K=\tilde{P}_\epsilon(-\x)$,  
which ends the proof.

}
\end{proof}

\vspace{0.5cm}

It is possible to construct an orthonormal basis in $H_\epsilon$ such that $Kf_s=f_s$. If $\Psi=\sum_s c_s f_s$ then  $K\Psi=\sum_s \overline{c_s} f_s$, i.e. $K$ becomes just a complex conjugation. Thus we can assume without loss of generality that the Hilbert space is
$\mathbb{C}^D$ with $2\leq D<\infty$ and $K$ is complex conjugation. In this case, all operators are characterized by finite dimensional matrices.

\subsection{Lifting continuity to real analyticity}

 The next lemma states that if
the map $P(\cdot)$ is real analytic, then a solution for Problem \ref{problema1} 
can always be turned into a solution for Problem \ref{problema2}.

\begin{lemma}\label{lema3}
Let $\R^d\ni \x\mapsto P(\x)\in \mathcal{L}(\mathbb{C}^D)$ be a
$\mathbb{Z}^d$-periodic, orthogonal projection valued map with ${\rm
  dim}({\rm Ran }P(\x))=N< D$ for all $\x\in \R^d$.  Assume that
all the components $\{P_{mn}(\cdot)\}_{1\leq m,n\leq D}$ in the standard basis of $\mathbb{C}^D$
have a $\mathbb{Z}^d$-periodic, bounded jointly analytic extension to a tubular complex neighborhood $I$ of $\R^d$. Assume
that $\{\Psi_j(\x)\}_{j=1}^N$ is an orthonormal basis of ${\rm Ran
}P(\x)$ for all $\x\in \R^d$, $\mathbb{Z}^d$-periodic, and
{\bf continuous}. Then there exists $\{\Xi_j(\cdot)\}_{j=1}^N$ {\bf holomorphic} on
a tubular complex neighborhood $\R^d\subset I'\subset I$,
$\mathbb{Z}^d$-periodic, which is an orthonormal basis of ${\rm Ran
}P(\x)$ if $\x\in \R^d$.
\end{lemma}

\begin{proof}
The function $g(\x)=\pi^{-d}\prod_{j=1}^d(1+k_j^2)^{-1}$ obeys  $g(\x)=g(-\x)$, it is  analytic on the strip $\{{\bf z}\in \mathbb{C}^d:\; |{\rm Im}(z_j)|<1,\; j\in\{1,...,d\}\}$, and $\int_{\R^d}
g(\x)d\x=1$. If $\delta>0$ we define an approximation of the Dirac
distribution $g_\delta(\x)=\delta^{-d}g(\delta^{-1}\x)$. Let us consider 
$$\Psi_{j,\delta}(\x)=\int_{\R^d}g_\delta(\x-\x')\Psi_j(\x')d\x',$$
which defines a $D$
dimensional vector which is $\mathbb{Z}^d$-periodic. 

{Let us show that $\Psi_{j,\delta}$ admits an analytic extension to a strip 
$$\{{\bf z}=(z_1,...z_d)\in \mathbb{C}^d:\; |{\rm Im}(z_j)|<\delta\}.$$  Due to 
periodicity, it is enough to find an extension to a bounded open neighbourhood $O$ of $[0,1]^d$ of the form
$$O=O_1\times ... \times O_d,\quad  O_j=\{z_j=x_j+iy_j\in\mathbb{C}:\; -2<x_j<2,\; |y_j|<\delta'<\delta\}.$$
For every ${\bf z}=(x_1+iy_1,...,x_d+iy_d)\in \overline{O}$ and every $\x'\in\R^d$ we have the estimate: 
$$|g_\delta({\bf z}-\x')|\leq \pi^{-d}\delta^d \prod_{j=1}^d\frac{ 1}{\delta^2-\delta'^2+(x_j-k_j')^2}.$$
Also:
$$g_\delta(\x-\x')=\frac{1}{(2\pi i)^d}\int_{\partial O_1}...\int_{\partial O_d}\frac{g_\delta(\mbox{\boldmath $\zeta$}-\x')}{\prod_{j=1}^d(\zeta_j-k_j)}d\mbox{\boldmath $\zeta$},\quad \x\in [0,1]^d,\quad \x'\in\R^d.$$

Because the $\Psi_j$'s are bounded on $\R^d$, the Fubini Theorem and the Cauchy Integral Formula allow one to analytically extend each $\Psi_{j,\delta}(\cdot)$ to $O$, and by periodicity, to a strip around $\R^d$.

Let us now continue the construction. A standard approximation argument implies:}
$$\lim_{\delta\to
  0}\left \{\max_{1\leq j\leq
    N}\sup_{\x\in\R^d}||\Psi_{j,\delta}(\x)-\Psi_{j}(\x)||\right \}=0.$$
Denote by $\Phi_{j,\delta}(\x):=P(\x)\Psi_{j,\delta}(\x)$. Then for a given $0<\epsilon\ll 1$ there exists $\delta$ small enough such that: 
$$ \max_{1\leq j\leq
    N}\sup_{\x\in\R^d}||\Phi_{j,\delta}(\x)-\Psi_{j}(\x)||<\epsilon. $$

Define the
  selfadjoint $N\times N$ Gram-Schmidt matrix
\begin{align}\label{dec7}
h_{kj}(\x)=\langle
\Phi_{j,\delta}(\x),\Phi_{k,\delta}(\x)\rangle=\int_{\R^{2d}}g_\delta(\x-\y)g_\delta(\x-\y')\langle
P(\x)\Psi_{j}(\y),\Psi_{k}(\y')\rangle d\y d\y'.
\end{align}

Since $\{\Psi_j(\x)\}_{j=1}^N$ is an orthonormal system, by choosing $\delta$ small enough we can insure that $||h(\x)-{\rm Id}||\leq 1/2$ uniformly for $\x\in\R^d$. 
Moreover, $h(\cdot)$ is bounded, $\mathbb{Z}^d$-periodic and real analytic, and the same holds true for $h(\cdot)^{-1/2}$. Then the vectors 
$$\Xi_j(\x):=\sum_{k=1}^N  \left [h(\x)^{-1/2}\right ]_{kj}\Phi_{k,\delta}(\x)$$
form a basis with all the required properties. 
\end{proof}

\subsection{Preservation of $CS$}

\begin{lemma}\label{lema3'}
Let $P(\x)$, $\{\Psi_j(\x)\}_{j=1}^N$ and $\{\Xi_j(\cdot)\}_{j=1}^N$
as constructed in Lemma \ref{lema3}. Assume that $P(\x)$ has the $CS$ property, and that $\overline{\Psi_j(\x)}=\Psi_j(-\x)$ for
all $j$. Then $\overline{\Xi_j(\x)}=\Xi_j(-\x)$ on $\R^d$ for all $j$.
\end{lemma}

\begin{proof}
We will refer to the already defined objects in the proof of Lemma
\ref{lema3}. First, 
due to the symmetry of $g$ it is easy to check that 
 $\overline{\Psi_{j,\delta}(\x)}=\Psi_{j,\delta}(-\x)$. This implies that  $\overline{\Phi_{j,\delta}(\x)}=\Phi_{j,\delta}(-\x)$. This leads to the identity $\overline{h_{jk}(\x)}=h_{jk}(-\x)$. Since $h(\x)^{-1/2}$ can be written as a power series in $h(\x)-{\rm Id}$, we have that $\overline{[h(\x)^{-1/2}]_{kj}}=[h(-\x)^{-1/2}]_{kj}$.  
\end{proof}

\subsection{Some facts about matrices having $CS$ or $CS '$ symmetry}

\begin{lemma}\label{lema-dec-1}
Let $\theta$ denote complex conjugation and consider a self-adjoint matrix family $\alpha(\x)\in \mathcal{L}(\mathbb{C}^N)$ having $CS$ symmetry, i.e. $\theta \alpha(\x)\theta =\alpha(-\x)$. Then:
\begin{equation}\label{dec-1}
[\alpha(\x)]_{mn}=[\alpha(-\x)]_{nm},\quad \sigma(\alpha(\x))=\sigma(\alpha(-\x)).
\end{equation} 
\end{lemma} 
\begin{proof} We know that $\overline{[\alpha(\x)]_{mn}}=[\alpha(\x)]_{nm}$ from self-adjointness and $\overline{[\alpha(\x)]_{mn}}=[\alpha(-\x)]_{mn}$ from $CS$ symmetry. The symmetry of the spectrum follows from the identity 
$${\rm det}(A-z\; {\rm Id})={\rm det}({}^tA-z\; {\rm Id}).$$
\end{proof}

\vspace{0.2cm}

\begin{lemma}\label{lema-dec-2}
Let $\theta$ denote the complex conjugation and consider a unitary matrix family $\beta(\x)\in \mathcal{L}(\mathbb{C}^N)$ with the $CS '$ symmetry, i.e. $\theta \beta(\x)\theta =\beta(-\x)^{-1}$ for all $\x$. Then:
\begin{equation}\label{dec-2}
[\beta(\x)]_{mn}=[\beta(-\x)]_{nm},\quad \sigma(\beta(\x))=\sigma(\beta(-\x)).
\end{equation} 
\end{lemma} 
\begin{proof} We know that $\overline{[\beta(-\x)^{-1}]_{mn}}=[\beta(-\x)]_{nm}$ from unitarity and $\overline{[\beta(-\x)^{-1}]_{mn}}=[\beta(\x)]_{mn}$ from the $CS '$ symmetry. The proof of the symmetry of the spectrum is the same as in the previous lemma.
\end{proof}

\vspace{0.5cm}

\begin{remark}\label{remark-dec}
We note the fact that both the $CS$ property for self-adjoint matrices and the $CS '$ property for unitary matrices boil down to the same identity, only involving  transposition: $${}^t\alpha(\x)=\alpha(-\x).$$
\end{remark}

\begin{lemma}\label{lemmma-dec-3}
Let $\x=[k_1,k_2]\in\R^2$ and let $\alpha(\x)$ be self-adjoint with  $CS$ symmetry. Then the unitary family $\beta(\x):=e^{ik_1 \alpha(\x)}$ has $CS$ symmetry, while the family $\gamma(\x):=e^{i\alpha(\x)}$ has  $CS '$ symmetry.
\end{lemma} 
\begin{proof} 
We have $$\theta \beta(\x) \theta =e^{-ik_1 \theta \alpha(\x)\theta }=e^{-ik_1 \alpha(-\x) }=\beta(-\x).$$
In a similar way: 
$$\theta \gamma(\x)\theta =e^{-i\theta \alpha(\x)\theta }=[\gamma(-\x)]^{-1}.$$
\end{proof}

\subsection{The case $d=1$}

Let us start by stating without proof a well-known fundamental lemma essentially due to Kato, formulated in a form which will be convenient for us in what follows.

\begin{lemma}\label{lema-dec-3}
Let $Q(x)=Q^*(x)=Q^2(x)$ be a $C^1$-norm differentiable family of orthogonal projections. Define 
\begin{equation}\label{III.1}
{ \mathcal{K}}(x):=i({\rm Id}-2Q(x))Q'(x)=i[Q'(x),Q(x)],\quad { \mathcal{K}}(x)={ \mathcal{K}}^*(x).
\end{equation}
Given $y\in\R$, one can define a family of unitary operators $A(x,y)$ such that 
\begin{equation}\label{III.3}
i\partial_x A(x,y)={ \mathcal{K}}(x)A(x,y),\quad A(y,y)={\rm Id},
\end{equation}
given by: 
\begin{equation}\label{III.3'}
A(x,y)={\rm Id}+\sum_{n\geq 1}(-i)^n\int_y^x ds_1\int_y^{s_1}ds_2...\int_y^{s_{n-1}}ds_n { \mathcal{K}}(s_1)...{ \mathcal{K}}(s_n).
\end{equation}
Furthermore, we have:
\begin{equation}\label{III.4}
i\partial_y A(x,y)=-A(x,y)H(y),\; [A(x,y)]^{-1}=A(y,x),\;
A(x_1,x_2)A(x_2,x_3)=A(x_1,x_3).
\end{equation}
The key intertwining identity is:
\begin{equation}\label{III.6}
Q(x)=A(x,x_0)Q(x_0)A(x_0,x).
\end{equation}
For fixed $x_0\in \R$, if either $Q(\cdot)\in C^n(\R)$ with $n\geq 1$ or $Q(\cdot)$ admits an analytic extension to the strip $\mathcal{I}_a^1:=\{z\in \mathbb{C}:\; |{\rm Im}(z)|<a\}$, then 
the same holds true for $A(\cdot,x_0)$ and $A(x_0,\cdot)$. 
\end{lemma}   

\vspace{0.5cm}

In the case in which the above family is $\Z$-periodic, i.e. $Q(x+1)=Q(x)$ for all $x\in \R$, then ${ \mathcal{K}}(x)$ is $\Z$-periodic and we observe that $A(x+1,x_0+1)$ must equal $A(x,x_0)$ (see \eqref{III.3}). Consider the unitary operator $A(1,0)$. Reasoning as in \eqref{septemb4}, the spectral theorem provides us with a self-adjoint operator $M$ with spectrum in $(-\pi,\pi]$ such that $A(1,0)=e^{i M}$. Also, due to the intertwining relation $Q(1)A(1,0)=A(1,0)Q(0)$ and because $Q(1)=Q(0)$ due to periodicity, we conclude that $A(1,0)$ commutes with $Q(0)$ and so does $M$.  

Define the unitary operator:
\begin{equation}\label{III.16}
U(k):=A(k,0)e^{-ikM},\quad k\in \R.
\end{equation}
One can verify that $U(\cdot)$ is $\Z$-periodic: 
$$U(k+1)=A(k+1,0)e^{-i(k+1)M}=A(k+1,1) A(1,0)e^{-i(k+1)M}=A(k,0)e^{-ikM}=U(k)$$
where we used $A(k+1,1)=A(k,0)$. 

Now let us consider Problem \ref{problema1} with $d=1$. Consider any orthonormal basis $\{\Psi_j(0)\}_{j=1}^N$ in the range of $P(0)$. Construct $U(k)$ as above starting from $P(k)$. We will now prove that the vectors:
\begin{equation}\label{solprob1}
\Xi_j(k):=U(k)\Psi_j(0),\quad k\in\R,\quad 1\leq j\leq N
\end{equation}
give a solution for Problem \ref{problema1}. These vectors are clearly periodic and orthogonal, thus we only need to prove that they live in the range of $P(k)$. But since $M$ commutes with $P(0)$ and due to the intertwining relation $P(k)A(k,0)=A(k,0)P(0)$, it follows that $P(k)U(k)=U(k)P(0)$ and we are done. 

We note that the above construction has not used a possible $CS$ property of 
$P(\cdot)$. If such a property does hold, i.e. besides smoothness and periodicity we also have $\theta P(k)\theta =P(-k)$ for all $k$, then  let us show that we can construct a basis which solves Problem \ref{problema3}. 

Since $\theta P(0)=P(0)\theta$, we have that $\theta \Psi_j(0)$ belongs to the range of $P(0)$, hence the real vectors $\frac{1}{2}(\Psi_j(0)+\theta \Psi_j(0))$ and $\frac{1}{2i}(\Psi_j(0)-\theta \Psi_j(0))$ have the same property. After a Gram-Schmidt procedure we obtain a real orthonormal basis of ${\rm Ran}(P(0))$, denoted by $\{\Xi_j(0)\}_{j=1}^N$. Thus the only thing we need to check is that $U(k)$ has the $CS$ property, i.e. $\theta U(k)\theta =U(-k)$. 

First, due to the identity $P'(-k)=-\theta P'(k)\theta$ we have $\theta { \mathcal{K}}(k)\theta ={ \mathcal{K}}(-k)$, thus $\theta A(k,0)\theta=A(-k,0)$ because they solve the same initial value problem. 

Second, taking $k=1$ in the last identity we get $\theta A(1,0)\theta =A(-1,0)=A(0,1)=[A(1,0)]^{-1}$ i.e.  
\begin{equation}\label{decem-11}
\theta e^{iM}\theta =e^{-iM},
\end{equation}
where $M=M^*$ commutes with $P(0)$ and has its spectrum in $(-\pi,\pi]$. Since the underlying Hilbert space is $\mathbb{C}^D$, all operators are represented by $D\times D$ matrices, hence $M$ has discrete eigenvalues and can be written as $M=\sum_{-\pi<\phi\leq \pi}\phi \Pi_\phi$ where $\Pi_\phi$ are spectral orthogonal projections corresponding to different eigenvalues. From \eqref{decem-11} we immediately obtain that $\theta \Pi_\phi\theta =\Pi_\phi$ for all eigenvalues, hence $\theta M\theta=M$. Thus:
$$\theta e^{-ikM}\theta =e^{ik\theta M\theta}=e^{ikM},\quad k\in \R,$$
hence $\theta U(k)\theta =U(-k)$ and the basis vectors we are looking for are given by $\Xi_j(k):=U(k)\Xi_j(0)$.

\begin{remark}\label{matricea-bate-tot}
Define $\Psi_m(k):=A(k,0)\Xi_m(0)$ and consider $\mathcal{M}_{jk}:=\langle M \Xi_k(0)|\Xi_j(0)\rangle$, i.e. the matrix elements of $P(0)MP(0)$ in the real basis $\Xi(0)$. Using the fact that $M$ commutes with $P(0)=P(1)$ leads to
\begin{equation}\label{matrix1}
\Psi_m(k+1)=A(k,0)P(0)e^{iM}P(0)\Xi_m(0)=\sum_{n=1}^N [e^{i\mathcal{M}}]_{nm}\Psi_n(k)
\end{equation}
and
\begin{equation}\label{matrix2}
\Xi_m(k)=U(k)\Xi_m(0)=A(k,0)P(0)e^{-ikM}\Xi_m(0)=\sum_{n=1}^N [e^{-ik\mathcal{M}}]_{nm}\Psi_n(k).
\end{equation}
In other words, we have found a unitary, $k$-dependent change of basis in the range of $P(k)$, which transforms the $\Psi(k)$'s into a periodic basis also preserving the $CS$ symmetry and the regularity of the original projection.  
\end{remark}

\subsection{The case $d=2$}
From now on (if not otherwise stated) we assume that the underlying Hilbert space is $\mathbb{C}^D$ and the family of projections $P(\cdot)$ has the $CS$ property, i.e. $\theta P(\x)\theta =P(-\x)$ for all $\x\in \R^2$. The idea is to proceed recursively. Consider $P(0,k_2)$. Using the result from $d=1$ we can construct an orthonormal basis $\{\Xi_j(0,k_2)\}_{j=1}^N$  of ${\rm Ran}P(0,k_2)$ such that 
\begin{equation}\label{III.25}
\Xi_j(0,k_2+1)=\Xi_j(0,k_2).\quad \theta \Xi_j(0,k_2)\theta =\Xi_j(0,-k_2),\quad \forall k_2\in \R.
\end{equation}
Now keeping $k_2$ fixed, consider the orthogonal projection family $\Pi(\cdot):=P(\cdot,k_2)$. We can apply Lemma \ref{lema-dec-3} with $Q$ replaced by $\Pi$ and obtain a family of unitary operators denoted by $A_{k_2}(x,y)$ generated by the self-adjoint operators 
$${ \mathcal{K}}_{k_2}(k_1):=i[\partial_1P(k_1,k_2),P(k_1,k_2)].$$
Reasoning as in the case $d=1$ we obtain the intertwining relation: 
\begin{equation}\label{III.26'}
P(k_1,k_2)A_{k_2}(k_1,0)=A_{k_2}(k_1,0)P(0,k_2), 
\end{equation}
and 
\begin{equation}\label{III.26}
\theta { \mathcal{K}}_{k_2}(k_1) \theta ={ \mathcal{K}}_{-k_2}(-k_1),\quad  
\theta A_{k_2}(k_1,0) \theta =A_{-k_2}(-k_1,0). 
\end{equation}
We also note the periodicity in $k_2$:
\begin{equation}\label{III.27}
 { \mathcal{K}}_{k_2+1}(k_1) ={ \mathcal{K}}_{k_2}(k_1),\quad  
 A_{k_2+1}(k_1,0) =A_{k_2}(k_1,0). 
\end{equation}
Due to the intertwining relation \eqref{III.26'} we have that the set of vectors 
$$\Psi_m(\x):=A_{k_2}(k_1,0)\Xi_m(0,k_2)$$
form an orthonormal basis in ${\rm Ran}P(\x)$, having all the right properties but one: they are not periodic in $k_1$. Using $A_{k_2}(k_1+1,1)=A_{k_2}(k_1,0)$ we  observe that:
$$\Psi_m(k_1+1,k_2)=A_{k_2}(k_1,0)A_{k_2}(1,0)\Xi_m(0,k_2).$$
From the intertwining relation \eqref{III.26'} and due to $P(1,k_2)=P(0,k_2)$ we have that $A_{k_2}(1,0)$ commutes with $P(0,k_2)$  hence:
\begin{equation}\label{matrice-dec}
\Psi_m(k_1+1,k_2)= \sum_{n=1}^N\langle A_{k_2}(1,0)\Xi_m(0,k_2)|\Xi_n(0,k_2)\rangle \Psi_n(k_1,k_2).
\end{equation}
The $N\times N$ matrix family defined by
\begin{equation}\label{matrice-dec-1}
\beta_{nm}(k_2):=\langle A_{k_2}(1,0)\Xi_m(0,k_2)|\Xi_n(0,k_2)\rangle
\end{equation}
is unitary in $\mathbb{C}^N$ and we have:
\begin{equation}\label{matrice-dec-2}
\Psi_m(k_1+1,k_2)= \sum_{n=1}^N\beta_{nm}(k_2)\; \Psi_n(k_1,k_2).
\end{equation}
This identity closely resembles \eqref{matrix1} but the important difference is that the matrix $\beta$ depends on $k_2$. The important question we have to answer is whether we can find a $\x$-dependent unitary rotation, like we did in \eqref{matrix2}, which transforms the $\Psi$'s into a $\Z^2$-periodic basis, preserving at least its continuity and $CS$ symmetry. { We note that $\beta(k_2)$ is similar to the `obstruction unitary' $U_{\rm obs}(k_2)$ introduced in \cite{FMP}.}

\begin{lemma}\label{lemma-dec-5}
The $N\times N$ matrix family defined by \eqref{matrice-dec-1} is unitary, $\Z$-periodic, inherits the regularity properties of $P(0,k_2)$ and has the $CS '$ symmetry.
\end{lemma}
\begin{proof}
We only prove the statement related to the $CS '$ symmetry. Since 
$$A_{-k_2}(1,0)=\theta A_{k_2}(-1,0)\theta=\theta A_{k_2}(0,1)\theta
=\theta [A_{k_2}(1,0)]^*\theta$$ 
and because $ \theta\Xi_m(0,-k_2)=\Xi_m(0,k_2)$ we have:
$$\beta_{nm}(-k_2)=\langle \theta [A_{k_2}(1,0)]^*\Xi_m(0,k_2)|\theta\Xi_n(0,k_2)\rangle =\langle \Xi_n(0,k_2)|[A_{k_2}(1,0)]^*\Xi_m(0,k_2)\rangle =\beta_{mn}(k_2),$$
which according to Lemma \ref{lema-dec-2} is equivalent with the $CS '$ symmetry.  
\end{proof}

\vspace{0.5cm} 

\begin{remark}\label{remarca4.2} Let us {\bf assume} that there exists an $N\times N$ self-adjoint  matrix family $h(k_2)$ which is continuous, $\Z$-periodic, with the $CS$ symmetry, such that $\beta(k_2)=e^{i h(k_2)}$. Note that this was the case when $d=1$ (there $h(k_2)=M$, see \eqref{matrix1}). Then the matrix $e^{-ik_1 h(k_2)}$ has the $CS '$ symmetry (see Lemma \ref{lemmma-dec-3}) and the vectors 
$$\Xi_m(\x):=\sum_{n=1}^N [e^{-ik_1 h(k_2)}]_{nm} \Psi_n(\x)$$
give the periodic, continuous and $CS$ symmetric basis we are looking for. Indeed, in this case we have:
$$\Xi_m(k_1+1,k_2)=\sum_{n=1}^N [e^{-i(k_1+1) h(k_2)}]_{nm} \sum_{j=1}^N[e^{i h(k_2)}]_{jn}\Psi_j(\x)=\Xi_m(\x).$$
\end{remark}

\vspace{0.5cm} 

Unfortunately, finding a `good'  logarithm for $\beta(k_2)$ (i.e. a matrix $h(k_2)$ which is self-adjoint, periodic, continuous, $CS$ symmetric and $\beta(k_2)=e^{ih(k_2)}$) is not always possible, as it can be seen from the following example. Let 
\begin{equation}\label{nologaritm}
\beta(k)=\left [ \begin{array}{ll}
         \quad \cos(2\pi k) & \sin(2\pi k)\\
        -\sin(2\pi k) & \cos(2\pi k)\end{array} \right ]=\cos(2\pi k){\rm Id}+i\sin(2\pi k)\sigma_2 ,\quad k\in\R
        \end{equation}
        be a family of unitary matrices, real analytic, $\Z$-periodic, and which obeys the $CS '$ symmetry. Note that $\beta(k)$ has degenerate eigenvalues at $k=0$ and $k=\pm1/2$, and the range of its spectrum covers the whole unit circle. 
        
        Let us show that one {\bf cannot} find a continuous self-adjoint family $h(k)$ which is also $\Z$-periodic, has the $CS$ symmetry and $\beta(k)=e^{ih(k)}$ for all $k$. Assume the contrary. From ${\rm det}(\beta(k))=1=e^{i {\rm Tr}(h(k))}$ we have that $k\mapsto \frac{1}{2\pi}{\rm Tr}(h(k))$ is a continuous, integer valued function. Hence there exists an integer $m$ such that ${\rm Tr}(h(k))=2\pi m$ for all $k$. Let $F_j(k):=\frac{{\rm Tr}(\sigma_j h(k))}{2}$; then we have
        $$h(k)=\frac{{\rm Tr}( h(k))}{2}{\rm Id}+\sum_{j=1}^3 F_j(k)\sigma_j=m\pi {\rm Id} +{\bf F}(k)\cdot \sigma.$$
        Thus 
        $$\beta(k)=(-1)^me^{i {\bf F}(k)\cdot \sigma}=(-1)^m \left(\cos(|{\bf F}(k)|){\rm Id}+i\frac{\sin(|{\bf F}(k)|)}{|{\bf F}(k)|}{\bf F}(k)\cdot \sigma \right).$$
        Comparing with the definition of $\beta(k)$, both $F_1$ and $F_3$ must be identically zero. Hence ${\bf F}(k)=[0,F_2(k),0]$ and 
        $|{\bf F}(k)|=|F_2(k)|$. Therefore
        $$\beta(k)=(-1)^m \left(\cos(F_2(k)){\rm Id}+i\sin(F_2(k))\sigma_2 \right)=\cos(F_2(k)+m\pi){\rm Id}+i\sin(F_2(k)+m\pi)\sigma_2.$$
        We conclude that there exists an integer-valued function $n(k)$ such that $F_2(k)+m\pi=2\pi k+2\pi n(k)$. Since $F_2$ must be continuous, it follows that $n(k)$ is constant and equals some $n$. Hence $F_2(k)=2\pi k+\pi(2n-m)$ for all $k$. In order to make sure that $\beta_{pq}(k)=\beta_{qp}(-k)$ we must have that $F_2(k)=-F_2(-k)$, i.e. $m=2n$, $(-1)^m=1$ and $F_2(k)=2\pi k$. This shows that $F_2$ cannot be periodic.

\vspace{0.5cm}

 Fortunately, the lack of a `good' logarithm can be circumvented  by a two step rotation procedure (see \eqref{zumba2}), where for each step a `good' logarithm can be constructed. In what follows, we will discuss four particular situations in which a good logarithm can be constructed and then show how they can be combined in order to deal with the general case.

\subsubsection{When $N=1$}
Here the unitary matrix $\beta(k_2)$ is just a complex number on the unit circle, which due to the $CS '$ property (see \eqref{dec-2}) is also even: $\beta(k_2)=\beta(-k_2)$.  
\begin{lemma}\label{lemaunu}
There exists a unique real { and even} function $\phi_p(k_2)=\phi_p(-k_2)$, with the same smoothness properties as $\beta$,  $\mathbb{Z}$-periodic, $\phi_p(0)=0$, such that $\beta(k_2)=\beta(0)e^{i\phi_p(k_2)}$. 
\end{lemma}
\begin{proof} Without loss of generality we may assume that $\beta(0)=1$. We start by first proving 
the existence of a continuous phase $\phi(t)$ such that $\beta(t)=e^{i\phi(t)}$ and 
$\phi(0)=0$. 

The map $\beta$ is uniformly continuous on $\R$, thus there exists $\delta>0$ such that 
$$|\beta(t)-\beta(s)|=|\beta(t)/\beta(s)-1|\leq 1/2 \quad {\rm whenever}\quad |t-s|\leq \delta.$$
If ${\rm Ln}$ denotes the principal branch of the natural logarithm, then in 
a neighbourhood $[-\delta,\delta]$ we can define $\phi(t)=-i{\rm Ln}(1+(\beta(t)-1))=\phi(-t)$ where $\beta(t)=e^{i\phi(t)}$, $\phi(0)=0$, $\phi$ is real and continuous. 

We can repeat this construction in the interval $(\delta,2\delta]$ by defining 
$$\phi(t)=\phi(\delta)-i{\rm Ln}[1+(\beta(t)/\beta(\delta)-1)].$$ 
A similar formula holds on $[-2\delta,-\delta)$, the extension being continuous at $\pm\delta$ and symmetric. After a finite number of steps we can cover any bounded interval of $\R$. Thus we obtain a continuous, symmetric, real function 
$\phi$ with $\phi(0)=0$, not necessarily periodic, such that $\beta(t)=e^{i\phi(t)}$. We only need to prove its $\Z$-periodicity. 

We know that $\beta$ is $\Z$-periodic, thus the function 
$$F(t):=\frac{\phi(t+1)-\phi(t)}{2\pi}$$
must be both continuous and integer valued. Thus $F(t)$ equals some constant $n\in \Z$. But we have $n=F(-1/2)=0$ because $\phi(1/2)=\phi(-1/2)$. This shows that $\phi(t+1)=\phi(t)$ for all $t$. 
The regularity properties of $\phi$ are the same as those of $\beta$, which can be seen from the formula
$$\phi(t)=\phi(t_0)-i{\rm Ln}[1+(\beta(t)/\beta(t_0)-1)]$$ 
valid for all $t$ with $|t-t_0|$ small enough. In particular, if $\beta$ is analytic in a strip containing $\R$, the same holds true for $\phi$.

The uniqueness of such a phase is a consequence of the following fact:  if  $\phi$ is continuous, $\phi(0)=0$, and $e^{i\phi(t)}=1$ for all $t$, then $\phi$ must be a constant equal to zero. 
\end{proof}

\subsubsection{When the eigenvalues of $\beta$ never cross}
Here we assume that $N>1$ and $\beta(k)$ has {\bf nondegenerate} spectrum for all $k$, besides being continuous, $\Z$-periodic and with $CS '$ symmetry.  
\begin{lemma}\label{lema-dec-10}
There exists a self-adjoint $N\times N$ matrix family $h(k)$ which is continuous and $\Z$-periodic, has the $CS$ property and  $\beta(k)=e^{ih(k)}$. 
\end{lemma}
\begin{proof}
{ $\beta$ is  uniformly continuous on $ [-1/2,1/2]$. Given any $\epsilon>0$ there exists $\delta>0$ such that $||\beta(k)-\beta(k')||\leq \epsilon$ whenever $|k-k'|\leq \delta$. Because  the eigenvalues never cross, there must exist some positive distance $\epsilon_0>0$ between any two eigenvalues at a given $k$. Let $\delta_0>0$ be such that $||\beta(k)-\beta(k')||\leq \epsilon_0/10$ whenever $|k-k'|\leq \delta_0$.

Because $\beta(0)$ has $N$ nondegenerate eigenvalues, we can order them according to their arguments. Thus we have $-\pi <\phi_1<\phi_2<...<\phi_N\leq \pi$ and $\lambda_j(0):=e^{i\phi_j}$. We know that $||\beta(k)-\beta(0)||\leq \epsilon_0/10$ whenever $|k|\leq \delta_0$. Regular perturbation theory implies that each set $\sigma(\beta(k))\cap B_{\epsilon_0/5}(\lambda_j(0))$ consists of exactly one point when $|k|\leq \delta_0$. The operators 
$$P_j(k):=\frac{i}{2\pi} \int_{|z-\lambda_j(0)|= \epsilon_0/5}(\beta(k)-z)^{-1}dz,\quad |k|\leq \delta_0,$$
are the orthogonal spectral projections of $\beta$ on $|k|\leq \delta_0$. They are continuous. Moreover,
$$\lambda_j(k):={\rm Tr} (\beta(k)P_j(k)),\quad |k|\leq \delta_0$$
are $N$ continuous eigenvalues. Also, because the spectrum of $\beta$ is symmetric (see \eqref{dec-2}), we must have $\lambda_j(k)=\lambda_j(-k)$ on $|k|\leq \delta_0$. Now we can repeat the procedure on $[-2\delta_0,-\delta_0]\cup [\delta_0,2\delta_0]$ starting from $k=\pm \delta_0$.  

After a finite number of steps we obtain some continuous eigenvalues $\lambda_j(k)=\lambda_j(-k)$ on $-1/2\leq k\leq 1/2$, together with their one dimensional Riesz projections.  The disk $|z-\lambda_j(k)|\leq \epsilon_0/5$ does not contain any other eigenvalues of $\beta(k)$ besides $\lambda_j(k)$, uniformly in $|k|\leq 1/2$ and $j$, and
$$P_j(k)=\frac{i}{2\pi} \int_{|z-\lambda_j(k)|= \epsilon_0/5}(\beta(k)-z)^{-1}dz,\quad |k|\leq 1/2.$$
Using the fact that ${}^t[(\beta(k)-z)^{-1}]=(\beta(-k)-z)^{-1}$ (from the $CS '$ property) together with $\lambda_j(k)=\lambda_j(-k)$, we have ${}^tP_j(k)={}^tP_j(-k)$. Moreover, $P_j(-1/2)=P_j(1/2)$ due to the periodicity of $\beta$.

Now any $k\in\R$ can be uniquely written as $k=n+\underline{k}$ with $n\in\Z$ and $-1/2<\underline{k}\leq 1/2$. Define 
$$\Lambda_j(k):=\lambda_j(\underline{k}),\quad P_j(k):=P_j(\underline{k}).$$ Since the spectrum of $\beta$ is periodic (as a set), the $\Lambda_j$'s are the periodic and continuous eigenvalues of $\beta$ on $\R$, which for $n>0$ obey:
$$\Lambda_j(n+ \underline{k})=\lambda_j(\underline{k})=\lambda_j(-\underline{k})=
\Lambda_j(-n-\underline{k}),$$
i.e. they are also even. From Lemma \ref{lemaunu} it follows that we can find $\phi_j$ periodic, continuous and even such that $\Lambda_j(k)=e^{i\phi_j(k)}$. The corresponding projections, $P_j(k)$, (as given by the
Riesz formula) are continuous, periodic, and obey
$$P_j(n+ \underline{k})=P_j(\underline{k})={}^tP_j(-\underline{k})=
{}^tP_j(-n-\underline{k}),\quad \forall n>0.$$
 Thus $h(k):=\sum_{j=1}^N \phi_j(k)P_j(k)$ obeys all the necessary conditions. 
 
 }
 
\end{proof}

\subsubsection{When a global Cayley transform exists}\label{subsection4.3}
Again we assume $N>1$, $\beta(\cdot)$ continuous, $\Z$-periodic, and with $CS '$ symmetry. We also assume that there exists some $\phi_0\in (-\pi,\pi]$ such that $e^{i\phi_0}$ {\bf never} belongs to the spectrum of $\beta(k)$. Then the statement of Lemma \ref{lema-dec-10} remains true, although the proof is rather different, as we will show in what follows. 

The matrix $\gamma(k):=e^{i(\pi-\phi_0)}\beta(k)$ has the property that $-1$ never lies in its spectrum. The matrix 
\begin{equation}\label{dec-14-1}
s(k):=i \{{\rm Id}-\gamma(k)\}\{{\rm Id}+\gamma(k)\}^{-1}
\end{equation}
is self-adjoint, continuous, periodic and $\theta s(k)\theta =s(-k)$, hence it has the $CS$ symmetry. By elementary algebra we get: 
$$\gamma(k)=\{{\rm Id}+is(k)\}\{{\rm Id}-is(k)\}^{-1}.$$
Consider $f(z)=(1+iz)/(1-iz)$ defined on a (narrow) strip $S$ containing the (real) spectrum of $s(k)$ for all $k$. Then:
\begin{equation}\label{dec-14-2}
h(k):=-\pi +\phi_0-i{\rm Ln} \gamma(k)=-\pi +\phi_0+\frac{1}{2\pi} \int_{\partial S} ({\rm Ln} f(z)) (s(k)-z)^{-1}dz 
\end{equation}
obeys all the properties.

\subsubsection{When $\beta(\cdot)$ is real analytic and $\beta(0)$ and $\beta(1/2)$ have nondegenerate spectrum}\label{subsection4.4}
This is the last particular situation in which we prove the existence of a `nice' logarithm with the properties listed in Lemma \ref{lema-dec-10}. The real analyticity of $\beta(\cdot)$ allows us to relax the strict non-crossing condition we had in Lemma \ref{lema-dec-10}, now we only demand nondegenerate spectrum at $k=0$ and $k=1/2$. Note that the periodicity of $\beta(\cdot)$ implies, in particular, that the spectrum at $k=-1/2$ is also nondegenerate.  

The first important remark is that the spectrum of $\beta(\cdot)$ consists of real analytic eigenvalues. Let us construct them. Assume that $e^{i\phi_0}$ is not in the spectrum of $\beta(0)$ and define 
$\gamma(k):=e^{i(\pi-\phi_0)}\beta(k)$. Then ${\rm Id}+\gamma(k)$ is invertible for $k$ near $0$ and we can proceed as in \eqref{dec-14-1} and \eqref{dec-14-2} and construct a self-adjoint matrix $h(\cdot)$ which is analytic in a small open disk containing $k=0$ such that 
$\beta(k)= e^{i(-\pi+\phi_0+h(k))}$ near $0$.  We can label the eigenvalues of $\beta(k)$ as $\{\lambda_j(k)=e^{i\phi_j(k)}\}_{j=1}^N$ choosing for example that at $k=0$ we have $-\pi< \phi_1(0)<...<\phi_N(0)\leq \pi$. 

{ The second step is to extend these projections and eigenvalues to $|k|\leq 1/2$ so that they remain real-analytic. If the eigenvalues of $\beta$ never cross, then the projections and eigenvalues we constructed in Lemma \ref{lema-dec-10} are the objects we are looking for. But here the eigenvalues can cross, and we need to be able to analytically continue the one dimensional projections through crossing points. By applying a finite number of times the Analytic Rellich Theorem \cite{R-S4} to the `local self-adjoint logarithms' \eqref{dec-14-2} we can follow analytically each projection $P_j(k)$ and each corresponding eigenvalue $\lambda_j(\cdot)$ to $|k|\leq 1/2$.  Because $\beta$ is periodic and $\beta(1/2)=\beta(-1/2)$ are non-degenerate, this construction can be extended to the whole of $\R$. 
}

Near $k=0$  the spectrum of $\beta$ is nondegenerate and we can write each $P_j(k)$ as a Riesz integral on a $k$-independent contour:
$$P_j(k)=\frac{i}{2\pi} \int_{|z-\lambda_j(0)|= r_0}(\beta(k)-z)^{-1}dz,$$
hence ${}^t[P_j(k)]=P_j(-k)$ near zero, identity which remains true on $\R$ through analytic continuation.

Also, because  $\lambda_j(k)=\lambda_j(-k)$ near $k=0$ (due to the $CS '$ symmetry and the non-degeneracy of the spectrum of $\beta(0)$), this equality remains true on $\R$ by analytic continuation.

The only thing we still need to prove is the periodicity of eigenvalues and their eigenprojections. We know that $\lambda_j(-1/2)=\lambda_j(1/2)=a_j$ from the symmetry, all of them being nondegenerate since by assumption $\beta(1/2)[=\beta(-1/2)]$ has nondegenerate spectrum. Let $r_0>0$ small enough such that $|z-a_j|\leq r_0$ does not contain other eigenvalues. Then if $k$ is small enough we have:
$$\frac{i}{2\pi} \int_{|z-a_j|= r_0}(\beta(-1/2+k)-z)^{-1}dz=\frac{i}{2\pi} \int_{|z-a_j|= r_0}(\beta(1/2+k)-z)^{-1}dz$$
due to the periodicity of $\beta$, thus $P_j(-1/2+k)=P_j(1/2+k)$ near $k=0$, hence they are equal everywhere through real analyticity. This proves that each $P_j$ is periodic. Finally, for small enough $k$ we have:
$$\lambda_j(-1/2+k)={\rm Tr}\{P_j(-1/2+k)\beta(-1/2+k)\}={\rm Tr}\{P_j(1/2+k)\beta(1/2+k)\}=\lambda_j(1/2+k)$$
which must hold everywhere so they are periodic. Now we can apply Lemma \ref{lemaunu} for each $\lambda_j$ obtaining $\lambda_j(k)=e^{i\phi_j(k)}$ where $\phi_j$ is continuous, $\Z$-periodic and even. Then the solution is $h(k):=\sum_{j=1}^N \phi_j(k)\;P_j(k)$.

\subsubsection{The general case}\label{atwostep}

Let us go back to \eqref{matrice-dec-2}, the starting point of our discussion. Remember that the matrix family $\beta(\cdot)$ has the $CS '$ symmetry and we want to unitarily rotate the $\Psi_m$'s so that the new basis stays $CS$ symmetric, remains at least continuous and is also periodic.  

Let us rewrite \eqref{matrice-dec-2} by setting $k_1$ equal to $-1/2$:
 \begin{equation}\label{dec-14-3}
\Psi_m(1/2,k_2)= \sum_{n=1}^N[\beta(k_2)]_{nm}\; \Psi_n(-1/2,k_2).
\end{equation}
We will see  in the next lemma that the only relevant interval is $k_1\in [-1/2,1/2]$.
\begin{lemma}\label{lemma-dec-16}
Assume that the vectors $\{\widetilde{\Xi}_n(x,k_2)\}_{n=1}^N$ form an orthonormal basis for $P(x,k_2)$ and are continuous on $[-1/2,1/2]\times \R$, have the $CS$ symmetry $\overline{\widetilde{\Xi}_n(x,k_2)}=\widetilde{\Xi}_n(-x,-k_2)$ and moreover,  $\widetilde{\Xi}_n(-1/2,k_2)=\widetilde{\Xi}_n(1/2,k_2)$. Given $k_1\in\R$, we can uniquely write it as $k_1=n+x$ with $n\in\Z$ and $x\in (-1/2,1/2]$. Then the vectors defined by 
$ \Xi_n(k_1,k_2):=\widetilde{\Xi}_n(x,k_2)$ are continuous on $\R^2$,  $CS$ symmetric and $\Z^2$-periodic.
\end{lemma}
\begin{proof}
The periodicity is a consequence of the identity  $\Xi_n(x+n,k_2)=\widetilde{\Xi}_n(x,k_2)$ which is valid for all $n\in\Z$ and $x\in (-1/2,1/2]$. The only thing we still need to prove is the $CS$ symmetry. We have:
$$\overline{\Xi_n(x+n,k_2)}=\overline{\widetilde{\Xi}_n(x,k_2)}
=\widetilde{\Xi}_n(-x,-k_2)=\Xi_n(-x-n,-k_2).$$
\end{proof}

\vspace{0.5cm} 

The main technical result needed for the proof of Theorem  \ref{metateorema} is the following:

\begin{proposition}\label{teorema-u}
There exists a family of unitary matrices $u(x,k_2)$ continuous on $[-1/2,1/2]\times \R$ and 
$\mathbb{Z}$-periodic in $k_2$, such that: 
\begin{equation}\label{nov10}
[u(-1/2,k_2)]^{-1}\beta(k_2)u(1/2,k_2)={\rm Id},
\end{equation}
which also has $CS$ symmetry:
\begin{equation}\label{nov11}
\theta u(x,k_2)\theta =u(-x,-k_2).
\end{equation}
Moreover, the vectors:
\begin{align}\label{nov4}
\widetilde{\Xi}_m(x,k_2):= \sum_{m=1}^N [u(x,k_2)]_{nm} \Psi_n(x,k_2)
\end{align}
obey the conditions needed in Lemma \ref{lemma-dec-16}.
\end{proposition}

\subsubsection{Proof of Proposition \ref{teorema-u}}

Before anything else, let us relate the proposition to the cases in which we can write $\beta(k_2)=e^{i h(k_2)}$ with some `nice' $h(k_2)$. In all those cases we can put $u(x,k_2)=e^{-ix h(k_2)}$ and we immediately see that 
\eqref{nov10} means no more than:
$$e^{-\frac{i}{2}h(k_2)}\beta(k_2)e^{-\frac{i}{2}h(k_2)}={\rm Id}.$$

Since a `good' logarithm might not always exist, the new important idea is to realize that one can straighten up the initial basis by performing {\bf two successive} unitary rotations.   First, given $\beta$ as in the theorem, we will show that we can always find some `good' (i.e. $\Z$-periodic, with $CS$ symmetry and continuous) $h(k_2)$ such that   
\begin{equation}\label{zumba1}
\sup_{k_2\in\R}||e^{-\frac{i}{2}h(k_2)}\beta(k_2)e^{-\frac{i}{2}h(k_2)}-{\rm Id}||\leq  1/2.
\end{equation}
Then rotating the $\Psi$-basis with $e^{-ix h(k_2)}$ we get  $\widetilde{\Psi}_m(x,k_2):= \sum_{n=1}^N [e^{-ix h(k_2)}]_{nm} \Psi_n(x,k_2)$ 
which obey: 
\begin{equation}\label{dec-14-4}
\widetilde{\Psi}_m(1/2,k_2)= \sum_{n=1}^N [e^{-\frac{i}{2}h(k_2)}\beta(k_2)e^{-\frac{i}{2}h(k_2)}]_{nm}\; \widetilde{\Psi}_n(-1/2,k_2).
\end{equation}
But now the new matching unitary matrix $$\tilde{\beta}(k_2):=e^{-\frac{i}{2}h(k_2)}\beta(k_2)e^{-\frac{i}{2}h(k_2)}={\rm Id}+e^{-\frac{i}{2}h(k_2)}[\beta(k_2)-e^{ih(k_2)}]e^{-\frac{i}{2}h(k_2)}$$ 
has the $CS '$ property, is periodic, continuous and always has $-1$ in its resolvent set. Reasoning as in Subsection \ref{subsection4.3} (see formula \eqref{dec-14-2}) we obtain another `good' $\tilde{h}(k_2)$ such that $\tilde{\beta}(k_2)=e^{i \tilde{h}(k_2)}$. By rotating $\widetilde{\Psi}_n(x,k_2)$ with $e^{-ix \tilde{h}(k_2)}$ we obtain our $\widetilde{\Xi}_n(x,k_2)$. This also gives an expression for the unitary $u$:
\begin{equation}\label{zumba2}
u(x,k_2)=e^{-ix h(k_2)}e^{-ix \tilde{h}(k_2)},\quad e^{-\frac{i}{2} \tilde{h}(k_2)}e^{-\frac{i}{2}h(k_2)}\beta(k_2)e^{-\frac{i}{2} h(k_2)}e^{-\frac{i}{2} \tilde{h}(k_2)}={\rm Id}.
\end{equation}

Now let us give the precise statements. 
\begin{lemma}\label{lemma-dec-17}
Assume that we can find $M\geq 1$ self-adjoint $N\times N$ matrix families $H_j(k_2)$ which are $\mathbb{Z}$-periodic, continuous and with $CS$ symmetry
so that:
\begin{align}\label{nov12}
e^{-\frac{i}{2} H_1(k_2)}\cdot \dots \cdot e^{-\frac{i}{2} H_M(k_2)}\beta(k_2)e^{-\frac{i}{2} H_M(k_2)}\cdot \dots \cdot e^{-\frac{i}{2} 
H_1(k_2)}={\rm Id}.
\end{align}
Then $u(x,k_2):=e^{-ix H_M(k_2)}...e^{-ix H_1(k_2)}$ is  unitary, $\mathbb{Z}$-periodic in the second variable and obeys \eqref{nov10} and \eqref{nov11}. Moreover, the vectors defined in \eqref{nov4} obey the conditions needed in Lemma \ref{lemma-dec-16}.
\end{lemma}
\begin{proof}
Unitarity of $u$ and its periodicity in $k_2$ are obvious, while the $CS$ symmetry is a consequence of Lemma \ref{lemmma-dec-3}. Also, \eqref{nov10} is a direct consequence of \eqref{nov12}. 

Now let us consider the vectors in \eqref{nov4}. The only thing left to be proved is that they are the same when $x=\pm 1/2$. We have: \begin{align}\label{nov5}
\Psi_m(x,k_2)= \sum_{s=1}^N [u(x,k_2)^{-1}]_{sm} \widetilde{\Xi}_s(x,k_2).
\end{align}
Then:
\begin{align}\label{nov6}
\widetilde{\Xi}_j(1/2,k_2)&= \sum_{r=1}^N [u(1/2,k_2)]_{rj} \Psi_r(1/2,k_2)=
\sum_{r=1}^N [u(1/2,k_2)]_{rj}\sum_{m=1}^N\beta_{mr}(k_2)\Psi_m(-1/2,k_2)\nonumber \\
& =\sum_{r=1}^N [u(1/2,k_2)]_{rj}\sum_{m=1}^N\beta_{mr}(k_2)\sum_{s=1}^N [u(-1/2,k_2)^{-1}]_{sm} \widetilde{\Xi}_s(-1/2,k_2)\nonumber \\
&=\sum_{s=1}^N [u(-1/2,k_2)^{-1}\beta(k_2)u(1/2,k_2)]_{sj}\widetilde{\Xi}_s(-1/2,k_2)
\nonumber \\
&=\widetilde{\Xi}_j(-1/2,k_2),
\end{align}
where in the last line we used \eqref{nov10}.
\end{proof}

\vspace{0.5cm}

The last difficult technical problem we still have to solve is the construction of the approximation given in \eqref{zumba1}. When this is done, \eqref{zumba2} provides the unitary we are looking for in Proposition \ref{teorema-u}, which finishes the construction. 

The following lemma solves the approximation problem and is of independent interest:

\begin{lemma}\label{propo-dec-1}
Let $\beta(k)$ be a {\bf continuous} family of unitary matrices, $\Z$-periodic and satisfying the $CS '$ property. Then there exists a sequence of {\bf real analytic} families $\beta_n(k)$  of unitary matrices, $\Z$-periodic, {\bf with nondegenerate spectrum at $0$ and $1/2$} and satisfying the $CS '$ symmetry such that:
\begin{align}\label{zumba3}
\lim_{n\to\infty}\sup_{k\in \R}||\beta_n(k)-\beta(k)||=0.
\end{align}
\end{lemma}
\begin{proof} Before giving the proof, let us remark that each $\beta_n$ satisfies the conditions from Subsection \ref{subsection4.4}, hence they can be written as $\beta_n(k)=e^{i h_n(k)}$ where each $h_n$ is continuous, periodic and with $CS$ symmetry. This implies \eqref{zumba1} since $||\beta_n(k)-\beta(k)||=||{\rm Id}-e^{-i h_n(k)/2}\beta(k)e^{-i h_n(k)/2}||$. 

Now let us prove the lemma. We split the proof into two parts. First, we construct a continuous family $\beta_n$, with a completely nondegenerate spectrum at $0$ and $1/2$, $\Z$-periodic and with $CS '$ symmetry. Second, we show how it can be made real analytic by preserving all other properties.  

\noindent {\bf Step 1.} Assume that the spectrum of $\beta(0)$ consists of $1\leq p_0\leq  N$ clusters of degenerate eigenvalues labeled as $\{\lambda_1(0),...,\lambda_{p_0}(0)\}$ in the increasing order of their principal arguments. If $s>0$ is sufficiently small and $|k|<s$, due to the continuity of $\beta(\cdot)$ we know that the spectrum of $\beta(k)$ will also consist of well separated 
clusters of eigenvalues. Let $C_j$ be a simple, positively oriented contour, independent of $|k|<s$, enclosing the $j$'th cluster and no other eigenvalues. Let $\Pi_j(k)$ be the spectral projection of $\beta(x)$ corresponding to the $j$'th cluster. We have:
\begin{align}\label{zumba4}
\Pi_j(k)=\frac{i}{2\pi}\int_{C_j}(\beta(k)-z)^{-1}dz,\quad {}^t \Pi_j(k)=\Pi_j(-k),\quad \lim_{k\to 0}||\Pi_j(k)-\Pi_j(0)||=0,\quad |k|<s.
\end{align}
The matrix $\beta(k)$ is block diagonal with respect to the decomposition $\mathbb{C}^N=\bigoplus_{j=1}^{p_0} \Pi_j(k)\mathbb{C}^N$, i.e. $\beta(k)=\sum_{j=1}^{p_0} \Pi_j(k)\beta(k)\Pi_j(k).$
Define $$\widetilde{\beta}(k):=\sum_{j=1}^{p_0} \lambda_j(0)\Pi_j(k)=e^{i\sum_{j=1}^{p_0} {\rm Arg}(\lambda_j(0))\Pi_j(k)},\quad |k|<s.$$
Clearly, $\widetilde{\beta}(k)$ is unitary, commutes with $\beta(k)$ and ${}^t\widetilde{\beta}(k)=\widetilde{\beta}(-k)$ if $|k|<s$. Define $\gamma(k):=\widetilde{\beta}^{-1}(k)\beta(k)$; we have that $\gamma(k)$ is unitary, commutes with $\beta(k)$,  ${}^t\gamma(k)=\gamma(-k)$ if $|k|<s$ and $\lim_{k\to 0}\gamma(k)={\rm Id}$. Using the Cayley transform as in Subsection \ref{subsection4.3}, see formulas \eqref{dec-14-1}-\eqref{dec-14-2} with $\phi_0=\pi$, we can find a self-adjoint matrix $\widetilde{h}(k)$ such that 
\begin{align}\label{zumba5}
\gamma(k)=e^{i \widetilde{h}(k)},\quad \widetilde{h}(0)=0,\quad {}^t\widetilde{h}(k)=\widetilde{h}(-k),\quad [\Pi_j(k),\widetilde{h}(k)]=0,\quad |k|<s.
\end{align}
Putting everything together and using that $\beta(k)=\beta(k+m)$ for all $m\in\Z$ we obtain:
\begin{align}\label{zumba6}
\beta(k+m)=e^{i(\sum_{j=1}^{p_0} {\rm Arg}(\lambda_j(0))\Pi_j(k) +\widetilde{h}(k))},\quad |k|<s,\quad \forall m\in\Z.
\end{align}
We proceed similarly for $k\in\{|k\pm 1/2|<s\}$. Notice that due to the periodicity of $\beta$, there will be the same number of eigenvalue clusters at $k=\pm 1/2$ and we can choose the same ordering for them, i.e. $\lambda_j(1/2)=\lambda_j(-1/2)$ for all $1\leq j\leq p_{1/2}\leq N$; note that there is no connection with the numbering at $k=0$. From the Riesz integral formula, choosing a contour $\tilde{C}_j$ around each $\lambda_j(1/2)=\lambda_j(-1/2)$ and using that $\beta(k)=\beta(k+1)$ we obtain $\Pi_j(-1/2)=\Pi_j(1/2)$ and moreover:
\begin{align}\label{zumba7}
{}^t\Pi_j(k)=\Pi_j(-k),\quad {}^t\Pi_j(1/2)=\Pi_j(-1/2)=\Pi_j(1/2),\quad ,\quad |k\pm 1/2|<s.
\end{align}
As in \eqref{zumba6} we get:
\begin{align}\label{zumba8}
\beta(k+m)=e^{i(\sum_{j=1}^{p_{1/2}} {\rm Arg}(\lambda_j(1/2))\Pi_j(k) +\widetilde{h}(k))},\quad |k\pm 1/2|<s,\quad \forall m\in\Z.
\end{align}
In particular, $\Pi_j(\pm 1/2)$ are real and symmetric matrices. 

The main idea is to locally perturb $\beta$ around all integers and half-integers such that the spectrum is nondegenerate at these points and all symmetry properties are left unchanged. Let us now describe the perturbation. First, consider $k=0$. Since $\Pi_j(0)^*=\Pi_j(0)={}^t\Pi_j(0)$, these matrices are also real and symmetric. Accordingly, ${\rm Ran}(\Pi_j(0))$ is invariant under complex conjugation and we can find a real orthonormal basis. Let $1\leq r_j(0)\leq N$ be the multiplicity of $\lambda_j(0)$, i.e. the dimension of   
${\rm Ran}(\Pi_j(0))$. We have:
\begin{align}\label{zumba9}
\Pi_j(0)=\sum_{l_j=1}^{r_j(0)} P_{j,l_j}(0),\quad {\rm dim\; Ran}(P_{j,l_j}(0))=1,\quad P_{j,l_j}=P_{j,l_j}^*={}^tP_{j,l_j}=P_{j,l_j}^2.
\end{align}
Define:
\begin{align}\label{zumba10}
A_j(0):=\sum_{l_j=1}^{r_j(0)} (l_j-1)\; P_{j,l_j}(0),\quad 1\leq j\leq p_0.
\end{align}
By construction, we have ${}^tA_j(0)=A_j(0)$. Seen as an operator acting on ${\rm Ran}(\Pi_j(0))$, the spectrum of $A_j(0)$ is nondegenerate and equals $\{0,1,...,r_j(0)-1\}$. 

Now consider $k=\pm 1/2$. Since $\Pi_j(1/2)=\Pi_j(-1/2)$ are real for all $1\leq j\leq p_{1/2}$ (see \eqref{zumba7}), we can mimic  the  construction in \eqref{zumba10} and obtain:
 \begin{align}\label{zumba11}
A_j(1/2):=\sum_{m_j=1}^{r_j(1/2)} (m_j-1)\; P_{j,m_j}(1/2),\quad 1\leq j\leq p_{1/2}.
\end{align}
Now let $g_s:\R\mapsto [0,1]$, continuous, even ($g_s(k)=g_s(-k)$), $g_s(0)=1$, ${\rm supp}(g_s)\subset [-s/2,s/2]$. Define:
\begin{align}\label{zumba11'}
v_s(k):=
s \sum_{m\in\Z} g_s(k+m) \left (\sum_{j=1}^{p_0}A_j(0)\right ) +
s \sum_{m\in\Z} g_s(k+m+1/2)\left (\sum_{j=1}^{p_{1/2}}A_j(1/2)\right ). 
\end{align}

From \eqref{zumba6} and \eqref{zumba8} we see that $\beta(k)=e^{ih(k)}$ on the set:
$$\Omega_s:=\bigcup_{m\in\Z} \{|k-m|<s\}\cup \{|k-m-1/2|<s\}.$$
We see that $\Omega_s$ is symmetric with respect to the origin and consists of a union of disjoint open intervals centered around all integers and half-integers. Also, the support of $v_s$ is contained in $\Omega_s$, it is periodic and has the $CS$ symmetry. Define $\beta_s(k)$ to be $ e^{i(h(k)+v_s(k))}$ if $k\in\Omega_s$, and let  $\beta_s(k)=\beta(k)$ outside $\Omega_s$. The family $\beta_s$ obeys all the required properties and converges uniformly in norm to $\beta$ when $s\to 0$. 

\noindent {\bf Step 2.} Here we will show how we can make $\beta_s$ real analytic without affecting the other required properties. In order to simplify notation, we drop the subscript $s$ and we assume that $\beta(\cdot)$ is continuous, with $CS '$ symmetry, $\Z$-periodic and with completely nondegenerate spectrum at $0$ and $1/2$ (hence at all integers and half-integers). Let $f(x)=\pi^{-1}(1+x^2)^{-1}$. If $\nu>0$ we define $f_\nu(x):=\nu^{-1}f(x/\nu)$. The function $f$ is positive, symmetric $f(x)=f(-x)$, $f_\nu$ is an approximation of the Dirac distribution and is analytic on the strip $\{|{\rm Im}(k)|<\nu\}$. Then the matrix family 
\begin{equation}\label{zumba12}
\mu_\nu(k):=\int_\R f_\nu (k-k') \beta(k')dk'
\end{equation}
is real analytic and $\Z$-periodic. A-priori $\mu_\nu(k)$ is neither self-adjoint nor unitary, but it obeys  ${}^t\mu_\nu(k)=\mu_\nu(-k)$. Also, $\mu_\nu(\cdot)$ converges uniformly to $\beta(\cdot)$ when $\nu\to 0$, which implies:
\begin{equation}\label{zumba13}
\lim_{\nu\to 0}\sup_{k\in\R}||\mu_\nu(k)\mu_\nu(k)^*-{\rm Id}||=0.
\end{equation}
If $\nu$ is small enough, the operator $\mu_\nu(k)\mu_\nu(k)^*$ is self-adjoint and positive, hence we can define:
\begin{equation}\label{zumba14}
\beta_\nu(k):=\{\mu_\nu(k)\mu_\nu(k)^*\}^{-1/2} \mu_\nu(k),\quad \beta_\nu(k)\beta_\nu(k)^*={\rm Id}.
\end{equation}
Clearly, $\beta_\nu(\cdot)$ is unitary, $\Z$-periodic and converges uniformly to $\beta(\cdot)$ when $\nu\to 0$.  This guarantees that $\beta_\nu$ has completely nondegenerate spectrum at $0$ and $\pm 1/2$ for $\nu$ small enough. We still need to prove that $\beta_\nu(\cdot) $ has the $CS '$ symmetry and it is real analytic. 

If $\nu$ is small enough we can write:
\begin{equation}\label{zumba15}
\beta_\nu(k)=\frac{i}{2\pi} \int_{|z-1|=1/2} z^{-1/2}(\mu_\nu(k)\mu_\nu(k)^*-z)^{-1}\mu_\nu(k)dz
\end{equation}
Now the matrix family  
$$\gamma_\nu(k):=\mu_\nu(k)\mu_\nu(k)^*=\int_{\R^2} f_\nu (k-k')f_\nu (k-k'') \beta(k')\beta(k'')^*dk'dk''$$
has a holomorphic extension to the strip $\{|{\rm Im}(k)|<\nu\}$. If $\nu$ is small enough, due to \eqref{zumba13} it follows that 
$(\mu_\nu(k)\mu_\nu(k)^*-z)^{-1}$ exists for all $|z-1|=1/2$ and is uniformly bounded on $\{|{\rm Im}(k)|<\alpha\}$ with $0<\alpha\ll \nu$. Together with \eqref{zumba12} this proves that $\beta_\nu$ has a holomorphic extension to $\{|{\rm Im}(k)|<\alpha\}$. 

The last thing we need to prove is ${}^t\beta_\nu(k)=\beta_\nu(-k)$, which would imply the $CS '$ property. From \eqref{zumba14} we see that if $\nu$ is small enough one can write:
$$ \beta_\nu(k)=\sum_{n\geq 0}a_n [\mu_\nu(k)\mu_\nu(k)^*-{\rm Id}]^n\mu_\nu(k),\quad (1+x)^{-1/2}=\sum_{n\geq 0} a_nx^n.$$
We have already seen that ${}^t\mu_\nu(k)=\mu_\nu(-k)$ and this leads to ${}^t[\mu_\nu(k)^*]=\mu_\nu(-k)^*$. Then:
$${}^t\{[\mu_\nu(k)\mu_\nu(k)^*-{\rm Id}]^n\mu_\nu(k)\}=\mu_\nu(-k)[\mu_\nu(-k)^*\mu_\nu(-k)-{\rm Id}]^n=[\mu_\nu(-k)\mu_\nu(-k)^*-{\rm Id}]^n\mu_\nu(-k)$$
where the second identity can be proved by induction with respect to $n$. The proof of the lemma is over.

\end{proof}

\section{Proof of Theorem \ref{teorema2}}\label{sec3}

We will use the generic notation $\alpha$ and $C$ for a finite number of positive constants which will appear during the proof. 

The difficulty stems from the fact that since $P_b$ is not norm continuous in $b$, one cannot use the continuity approach of \cite{NN2}. The first step in overcoming this difficulty is to use an old ansatz by Peierls \cite{Pe} and Luttinger \cite{Lu} and to couple it with the regularized magnetic perturbation theory in order to construct (out of the $w_{j,\gamma}$'s) an intermediate orthogonal projection $\Pi_b$ such that 
\begin{equation}\label{M.14}
\lim_{b\to 0} ||\Pi_b-P_b||=0.
\end{equation}
The second step is to construct a localized orthonormal basis for ${\rm Ran}(\Pi_b)$, while the last step is to transfer it to ${\rm Ran}(P_b)$ by the `continuity' method of \cite{N2, NN2}.

If $|b|$ is small enough, the orthogonal projection $P_b$ can be written as a Riesz integral (see \eqref{M-8}): 
$$P_b=\frac{i}{2\pi} \int_\mathcal{C} (H_b-z)^{-1}dz$$
where $\mathcal{C}$ is a simple, positively oriented contour which surrounds $\sigma_0$ (hence $\sigma_b$) and no other part of the spectrum of $H_0$. By writing
$$P_b=\frac{i}{2\pi} \int_\mathcal{C} [(H_b-z)^{-1}-(H_b-i)^{-1}]dz=
\frac{i}{2\pi} \int_\mathcal{C} (z-i)(H_b-z)^{-1}(H_b-i)^{-1}dz$$
{ and using \eqref{M-8}}, one can prove ({ see also Section 4.1 and 4.2 in \cite{C}}) that $P_b$ has a jointly continuous integral kernel $P_b(\xy,\xy')$ and there exist $\alpha>0$ and $C<\infty$ such that
\begin{equation}\label{M.15}
|P_b(\xy,\xy')|\leq C  e^{-\alpha||\xy-\xy'||},\quad \forall \xy,\xy'\in\R^2,\quad |b|\leq b_0.
\end{equation}
In particular, this holds true for $b=0$. Since $w_{j,\gamma}(\xy)=\int_{\R^2} P_0(\xy,\xy')w_{j,\gamma}(\xy')d\xy'$,  we have that  each $w_{j,\gamma}$ is continuous and by applying the Cauchy-Schwarz inequality and \eqref{M.15} we obtain the existence of some constants $C'<\infty$ and $\alpha'<\alpha$ such that
\begin{equation}\label{M.16}
\sup_{\xy\in\R^2,\gamma\in \Gamma}e^{\alpha' ||\xy-\gamma||}|w_{j,\gamma}(\xy)|\leq C'.
\end{equation}
This implies the identity:
\begin{equation}\label{M.16'}
P_0(\xy,\xy')=\sum_{j=1}^N\sum_{\gamma\in \Gamma} w_{j,\gamma}(\xy)\overline{w_{j,\gamma}}(\xy')
\end{equation}
where the above series is absolutely convergent; let us show why. The estimate \eqref{M.16} can be reinterpreted as $|w_{j,\gamma}(\xy)|\leq Ce^{-\alpha ||\xy-\gamma||}$ for some $C<\infty$ and $\alpha>0$, uniformly in $\gamma$ and $j$. Using the triangle inequality we obtain:
\begin{align}\label{ianuar12''}
\sum_{j=1}^N\sum_{\gamma\in \Gamma} |w_{j,\gamma}(\xy)|\; |w_{j,\gamma}(\xy')|\leq {\rm const}\;  e^{-\alpha ||\xy-\xy'||/2} .
\end{align}
 Later we will need the following bound which holds true due to our sparsity assumption on $\Gamma$: 
 \begin{align}\label{ianuar12'''}
\sup_{\bc\in\R^2}\sum_{\gamma\in \Gamma} ||\bc-\gamma||^k e^{-\alpha ||\bc-\gamma||} <\infty ,\quad \forall k \geq 0.
\end{align}

\subsection{Proof of {\rm (i)}.}
Applying the Riesz integral formula in \eqref{M-8} we obtain:
\begin{equation}\label{ianuar11}
\left |P_b(\xy,\xy')-e^{ib\phi(\xy,\xy')}P_0(\xy,\xy')\right |\leq C\;|b| \;e^{-\alpha ||\xy-\xy'||}.
\end{equation}
Define:
\begin{equation}\label{M.17}
\widehat{P}_b(\xy,\xy'):=e^{ib\phi(\xy,\xy')}P_0(\xy,\xy'),
\end{equation}
and denote by $\widehat{P}_b$ the corresponding operator. 
{ Applying a Schur-Holmgren estimate using \eqref{ianuar11} leads us to:
\begin{equation}\label{ianuar15}
 ||\widehat{P}_b-{P}_b||\leq C\;|b|,\quad {\rm where}\quad |{P}_b(\xy,\xy')-\widehat{P}_b(\xy,\xy')|\leq C\;|b|\; e^{-\alpha ||\xy-\xy'||}.
\end{equation}
}
Define the kernel: 
\begin{equation}\label{ianuar12}
\widetilde{P}_b(\xy,\xy')=\sum_{j=1}^N\sum_{\gamma\in \Gamma} e^{ib\phi(\xy,\gamma)}w_{j,\gamma}(\xy)\overline{e^{ib\phi(\xy',\gamma)}w_{j,\gamma}(\xy')},
\end{equation}
and by $\widetilde{P}_b$ the corresponding operator. The series is absolutely convergent (see the estimate in \eqref{ianuar12''}) and defines a jointly continuous function. Moreover, we have:
\begin{equation}\label{ianuar14'}
|\widetilde{P}_b(\xy,\xy')|\leq C\; e^{-\alpha ||\xy-\xy'||}.
\end{equation}
From the definition of $\phi(\xy,\xy')$ (see \eqref{M-9} and \eqref{M-5}) it follows that it equals the magnetic flux through the triangle with corners situated at $0$, $\xy$ and $\xy'$.

{ Let $fl(\xy,\y,\xy')$ denote the magnetic flux through  the triangle with corners situated at $\xy$, $\y$ and $\xy'$. }
Using Stokes' theorem and the fact that $|B(\xy)|\leq 1$ we have:
\begin{equation}\label{ianuar9'}
fl(\xy,\y,\xy')=\phi(\xy,\y)+\phi(\y,\xy')-\phi(\xy,\xy'),\quad |fl(\xy,\y,\xy')|\leq \frac{1}{2}|\xy-\y|\; |\y-\xy'|.
\end{equation}
Using the various definitions, the triangle and Cauchy-Schwarz inequality we obtain:
\begin{align*}
|\widehat{P}_b(\xy,\xy')-\widetilde{P}_b(\xy,\xy')|&\leq \sum_{j=1}^N\sum_{\gamma\in \Gamma} 2|\sin(b \;fl(\xy,\gamma,\xy')/2)| \;
|w_{j,\gamma}(\xy)|\; |w_{j,\gamma}(\xy')|\\
&\leq |b| N C^2 e^{-\alpha ||\xy-\xy'||/2}\sup_{\bc\in\R^2}\sum_{\gamma\in \Gamma} ||\bc-\gamma||^2 e^{-\alpha ||\bc-\gamma||},
\end{align*}
which leads to the existence of a $C<\infty$ and $\alpha>0$ such that 
\begin{equation}\label{ianuar14}
|\widehat{P}_b(\xy,\xy')-\widetilde{P}_b(\xy,\xy')|\leq C\;|b|\; e^{-\alpha ||\xy-\xy'||},\quad ||\widehat{P}_b-\widetilde{P}_b||\leq C\;|b|.
\end{equation}

Consider the closed subspace 
$$S_b:={\rm closure}\{{\rm Span}\{e^{i b\phi(\cdot,\gamma)}w_{j,\gamma}(\cdot):\; \gamma\in \Gamma,\; j\in\{1,2,...,N\}\}\}$$
and its associated Gram-Schmidt matrix 
\begin{equation}\label{M.28}
M_b(\gamma,j;\gamma',j'):=\left \langle e^{i b\phi(\cdot,\gamma')}w_{j',\gamma'}(\cdot)|e^{i b\phi(\cdot,\gamma)}w_{j,\gamma}(\cdot)\right \rangle_{L^2(\R^2)}.
\end{equation}
\begin{lemma}\label{lemaianuar1}
Consider the self-adjoint operator $M_b$ acting on the space $l^2(\Gamma)\otimes \C^N$ given by the matrix elements $M_b(\gamma,j;\gamma',j')$. 
If $b_0$ is small enough, then there exists some $\alpha>0$ and $C<\infty$ such that if $|b|\leq b_0$:
\begin{equation}\label{yanuarie1}
|\delta_{jj'}\delta_{\gamma \gamma'}-M_b(\gamma,j;\gamma',j')|\leq C\;|b|\; e^{-\alpha||\gamma-\gamma'||}.
\end{equation}
Moreover, $M_b\geq 1/2$, and if $\beta\in\{1/2,1\}$ then  
\begin{equation}\label{yanuarie1'}
|\delta_{jj'}\delta_{\gamma \gamma'}-M_b^{-\beta}(\gamma,j;\gamma',j')|\leq C\;|b|\; e^{-\alpha||\gamma-\gamma'||}.
\end{equation}
Finally, the vectors:
\begin{equation}\label{ianuar21}
\Psi_{j,\gamma,b}(\xy):=\sum_{\gamma'\in\Z^2}\sum_{j'=1}^N[M_b^{-1/2}](\gamma',j';\gamma,j)e^{i b\phi(\xy,\gamma')}w_{j',\gamma'}(\xy)
\end{equation}
form an orthonormal basis of $S_b$, being uniformly exponentially localized around $\gamma$, i.e. 
\begin{equation}\label{febrya}
|\Psi_{j,\gamma,b}(\xy)|\leq C\; e^{-\alpha||\xy-\gamma||},\; 
|\Psi_{j,\gamma,b}(\xy)-e^{i b\phi(\xy,\gamma)}w_{j,\gamma}(\xy)|\leq C\; |b|\; e^{-\alpha||\xy-\gamma||}.
\end{equation}
\end{lemma}
\begin{proof}
Define $D_b:=M_b-{\rm Id}$. Since $ \delta_{jj'}\delta_{\gamma \gamma'}=\left \langle w_{j',\gamma'}(\cdot)|w_{j,\gamma}(\cdot)\right \rangle_{L^2(\R^2)}$ and $\phi(\gamma,\gamma)=0$ we have: 
\begin{equation*}
D_b(\gamma,j;\gamma',j')= e^{ib\phi(\gamma',\gamma)} \int_{\R^d} w_{j',\gamma'}(\xy)\; \overline{w_{j,\gamma}(\xy)}\left (e^{ib\; fl(\gamma',\xy,\gamma)}-1\right )d\xy
\end{equation*}
and
$$|D_b(\gamma,j;\gamma',j')|\leq |b| e^{-\alpha'||\gamma-\gamma'||}\int_{\R^d}||\gamma'-\xy||  e^{\alpha'||\gamma'-\xy||}
|w_{j',\gamma'}(\xy)|\; ||\gamma-\xy|| \; e^{\alpha'||\gamma-\xy||}
|w_{j,\gamma}(\xy)| d\xy$$
for some small enough $\alpha'>0$. Using the Cauchy-Schwarz  inequality we find:
\begin{equation}\label{febry1}
|D_b(\gamma,j;\gamma',j')|\leq |b| \; C e^{-\alpha||\gamma-\gamma'||}.
\end{equation}
If $n\geq 2$ we have
$$D_b^n(\gamma,j;\gamma',j')=\sum_{j_1,\gamma_1}...\sum_{j_{n-1},\gamma_{n-1}}D_b(\gamma,j;\gamma_1,j_1)...D_b(\gamma_{n-1},j_{n-1};\gamma',j')$$
and \eqref{febry1} implies:
\begin{align}\label{febry2}
e^{\alpha||\gamma-\gamma'||}|D_b^n(\gamma,j;\gamma',j')|&\leq |b|CN^{n-1} \left (\sup_{k, j'',\gamma''}\sum_{\bc\in \Gamma}|D_b(\gamma'',j'';\bc,k)|e^{\alpha||\gamma''-\bc||}\right )^{n-1}\nonumber \\
&\leq 
|b|^n \tilde{C}^n,\quad \forall n\geq 1.
\end{align}
The estimate \eqref{febry1} also implies that $D_b$ goes to zero with $b$ in the operator norm, hence $M_b$ is close to the identity operator, thus  $M_b^{-\beta}=({\rm Id}+D_b)^{-\beta}$ exists and can be expressed as a norm-convergent power series around zero.  Then \eqref{febry2}  implies \eqref{yanuarie1'}. Finally, it is straightforward to check that the vectors defined in \eqref{ianuar21} are orthogonal and span $S_b$, while \eqref{febrya} is an easy consequence of \eqref{yanuarie1'} and the triangle inequality. 
\end{proof}

\vspace{0.5cm}

The orthogonal projection associated to the subspace $S_b$ is 
\begin{equation}\label{ianuarie10}
\Pi_b:=\sum_{\gamma\in\Gamma}\sum_{j=1}^N |\Psi_{j,\gamma,b}\rangle  \langle \Psi_{j,\gamma,b}|.
\end{equation}
Under the conditions of Lemma \ref{lemaianuar1}, using \eqref{ianuar21} and the self-adjointness of $M_b$ we obtain that $\Pi_b$ has a jointly continuous integral kernel given by:
$$\Pi_b(\xy,\xy')=\sum_{\gamma',\gamma''\in\Gamma}\sum_{j',j''=1}^N M_b^{-1}(\gamma',j';\gamma'',j'')e^{i b\phi(\xy,\gamma')}w_{j',\gamma'}(\xy)e^{i b\phi(\gamma'',\xy')}\overline{w_{j'',\gamma''}(\xy')}$$
where the series converges absolutely. Using  \eqref{yanuarie1'} and \eqref{ianuar12} we find that there exists some $\alpha>0$ such that  
\begin{equation}\label{yanuarie2}
|\Pi_b(\xy,\xy')-\widetilde{P}_b(\xy,\xy')|\leq  C\;|b|\; e^{-\alpha||\xy-\xy'||}.
\end{equation}
Using \eqref{ianuar15} and \eqref{ianuar14} we obtain:
\begin{equation}\label{yanuarie3}
|\Pi_b(\xy,\xy')-{P}_b(\xy,\xy')|\leq  C\;|b|\; e^{-\alpha||\xy-\xy'||},\quad ||\Pi_b-P_b||\leq C\;|b|,
\end{equation}
which finishes the proof of \eqref{M.14}. 

The next step is to consider the Sz.-Nagy unitary operator $U_b$ which intertwines $P_b$ with $\Pi_b$, i.e. $P_bU_b=U_b\Pi_b$, given (as in \eqref{dech11}) by the formula: 
\begin{equation}\label{dech11'}
U_b:=\{P_b\Pi_b +({\rm Id}-P_b)({\rm Id}-\Pi_b)\}
\{{\rm Id}-(\Pi_b-P_b)^2\}^{-1/2}.
\end{equation}
A consequence of \eqref{yanuarie3} is that for small enough $b$, the operator $U_b-{\rm Id}$ has a jointly continuous integral kernel which obeys the same estimate as in \eqref{yanuarie3}.  
Then the set of vectors:
\begin{equation}\label{finalian}
\Xi_{j,\gamma,b}:=U_b\Psi_{j,\gamma,b}=
\sum_{\gamma'\in\Gamma}\sum_{j'=1}^N M_b^{-1/2}(\gamma',j';\gamma,j)[U_b \; e^{i b\phi(\cdot,\gamma')}w_{j',\gamma'}(\cdot)]
\end{equation}
form an orthonormal basis of ${\rm Ran}(P_b)$ satisfying \eqref{M-10} and \eqref{M-11}. 

\vspace{0.5cm}

\subsection{Proof of {\rm (ii)}.}
Remember that here the magnetic field is constant with magnitude $B(\xy)=1$ and $\Gamma$ is a periodic lattice embedded in $\R^2$. In this case, the phase is not just antisymmetric but it is also bilinear, since now $\phi(\xy,\xy')= \frac{1}{2}(x_1'x_2-x_1x_2')$ or, with the usual abuse of language, $\phi(\xy,\xy')=\frac{1}{2}{\bf B}\cdot (\xy'\wedge \xy)$. 
The identity $(\xy-\y)\wedge (\xy'-\y)=-\xy'\wedge \xy +\y\wedge \xy +\xy'\wedge \y$ 
implies:
\begin{equation}\label{ianuar20}
\phi(\xy,\y)+\phi(\y,\xy')-\phi(\xy,\xy')=\phi(\xy'-\y,\xy-\y)=\phi(\xy'-\y,\xy-\xy').
\end{equation}
Denote by $\tau_\gamma^b$ the unitary magnetic translation given by:
$$[\tau_\gamma^bf](\xy)=e^{ib\phi(\xy,\gamma)}f(\xy-\gamma).$$ 

Using \eqref{M-12}, \eqref{ianuar20} and \eqref{M.28} we obtain:
$$M_b(\gamma,j;\gamma',j')=\langle \tau_{\gamma'}^b (w_{j',0})|\tau_{\gamma}^b ( w_{j,0})\rangle_{L^2(\R^2)}=:e^{ib \phi(\gamma,\gamma')}\mathcal{M}_b(\gamma-\gamma';j,j'),$$
where 
\begin{equation}\label{febry15}
\mathcal{M}_b(\gamma'';j,j')=\int_{\R^2} e^{-ib\phi(\y,\gamma'')}\overline{w_{j,0}(\y-\gamma'')}\;w_{j',0}(\y)d\y=\langle w_{j',0}| \tau_{\gamma''}^b (w_{j,0})\rangle_{L^2(\R^2)}.
\end{equation}
\begin{lemma}\label{lemafebry} There exists some $b_0$ small enough such that the operator $M_b^{-1/2}$ has a matrix of the form
$$M_b^{-1/2}(\gamma,j;\gamma',j')=:e^{ib \phi(\gamma,\gamma')}\mathcal{T}_b(\gamma-\gamma';j,j'),$$
where $|\mathcal{T}_b(\gamma'';j,j')|\leq Ce^{-\alpha |\gamma''|}$ for every $0\leq b\leq b_0$. 
\end{lemma}
\begin{proof}
In the proof of Lemma \eqref{lemaianuar1} we introduced the operator $D_b=M_b-{\rm Id}$ in $l^2(\Gamma)\otimes \C^N$ and showed that if $|b|$ is small enough then 
 $M_b^{-1/2}=({\rm Id}+D_d)^{-1/2}$ exists as a power series in $D_b$. The matrix elements of $D_b$ inherit the symmetry properties of $M_b$:
  $$\mathcal{D}_b(\gamma'';j,j'):=\mathcal{M}_b(\gamma'';j,j')-
  \delta_{\gamma''0}\delta_{jj'},\quad D_b(\gamma,j;\gamma',j')=e^{ib \phi(\gamma,\gamma')}\mathcal{D}_b(\gamma-\gamma';j,j').$$ 
  Then using \eqref{ianuar20} and the fact that $\Gamma$ is a lattice we obtain:
  \begin{align}\label{febry10}
 D_b^2(\gamma,j;\gamma',j')&=\sum_{j''=1}^N \sum_{\gamma''\in\Gamma} e^{ib (\phi(\gamma,\gamma'')+\phi(\gamma'',\gamma'))} \mathcal{D}_b(\gamma-\gamma'';j,j'')
 \mathcal{D}_b(\gamma''-\gamma';j'',j')\nonumber \\
 &=e^{ib \phi(\gamma,\gamma')}
 \sum_{j''=1}^N \sum_{\gamma''\in\Gamma}e^{ib \phi(\gamma'-\gamma'',\gamma-\gamma')}\mathcal{D}_b(\gamma-\gamma'';j,j'')
 \mathcal{D}_b(\gamma''-\gamma';j'',j')\nonumber \\
 &=:e^{ib \phi(\gamma,\gamma')}\mathcal{D}_b^{(2)}(\gamma-\gamma';j,j'),
  \end{align}
where 
$$\mathcal{D}_b^{(2)}(\tilde{\gamma};j,j')=\sum_{j''=1}^N \sum_{\gamma''\in\Gamma}e^{ib \phi(\tilde{\gamma},\gamma'')}\mathcal{D}_b(\tilde{\gamma}-\gamma'';j,j'')
 \mathcal{D}_b(\gamma'';j'',j').$$
If $n\geq 3$ one can prove by induction that $D_b^n(\gamma,j;\gamma',j'):=e^{ib \phi(\gamma,\gamma')}\mathcal{D}_b^{(n)}(\gamma-\gamma';j,j')$ where:
$$\mathcal{D}_b^{(n)}(\tilde{\gamma};j,j')=\sum_{j''=1}^N \sum_{\gamma''\in\Gamma}e^{ib \phi(\tilde{\gamma},\gamma'')}\mathcal{D}_b^{(n-1)}(\tilde{\gamma}-\gamma'';j,j'')
 \mathcal{D}_b(\gamma'';j'',j').$$
 This structure will be preserved for any convergent series in $D_b$, while the exponential localization is a direct consequence of \eqref{yanuarie1'}.
 
\end{proof} 

\vspace{0.5cm}

Using the above expression for $M_b^{-1/2}(\gamma,j;\gamma',j')$  in \eqref{ianuar21} we obtain:
\begin{align}\label{ianuar22}
\Psi_{j,\gamma,b}(\xy)=\sum_{\gamma'\in\Gamma}\sum_{j'=1}e^{ib \phi(\gamma',\gamma)}\mathcal{T}_b(\gamma'-\gamma;j',j)e^{ib \phi(\xy,\gamma')} w_{j',0}(\xy-\gamma').
\end{align}
We note the identity 
$$\phi(\gamma-\gamma',\xy-\gamma')=\phi(\xy-\gamma,\gamma'-\gamma)$$
which together with \eqref{ianuar20} gives:
$$\phi(\xy,\gamma')+\phi(\gamma',\gamma)-\phi(\xy,\gamma)=\phi(\gamma-\gamma',\xy-\gamma')=\phi(\xy-\gamma,\gamma'-\gamma).$$
Replacing this in \eqref{ianuar22} we obtain:
\begin{align}\label{ianuar23'}
\Psi_{j,\gamma,b}(\xy)&=e^{ib \phi(\xy,\gamma)}\sum_{\gamma'\in\Gamma}\sum_{j'=1}^N
\mathcal{T}_b(\gamma'-\gamma;j',j)e^{ib \phi(\xy-\gamma,\gamma'-\gamma)}w_{j',0}(\xy-\gamma')\nonumber \\
&=e^{ib \phi(\xy,\gamma)}\sum_{\gamma''\in\Gamma}\sum_{j'=1}^N
\mathcal{T}_b(\gamma'';j',j)e^{ib \phi(\xy-\gamma,\gamma'')} w_{j',0}(\xy-\gamma-\gamma'')\nonumber \\
&=\left [\tau_{\gamma}^b\left (\sum_{\gamma''\in\Gamma}\sum_{j'=1}^N
\mathcal{T}_b(\gamma'';j',j)e^{ib \phi(\cdot,\gamma'')} w_{j',0}(\cdot-\gamma'')\right )\right ](\xy).
\end{align}
Denote by:
\begin{equation}\label{febry14}
u_{j,b}:=\sum_{\gamma''\in\Gamma}\sum_{j'=1}^N
\mathcal{T}_b(\gamma'';j',j)\;\tau_{\gamma''}^b (w_{j',0}).
\end{equation}
From \eqref{ianuar23'} we have $\Psi_{j,\gamma,b}=\tau_{\gamma}^b (u_{j,b})$. The last technical result we need is the following: 
\begin{lemma}
Both $P_b$ and $\Pi_b$ commute with the magnetic translations $\tau_{\gamma}^b$, and so does $U_b$.
\end{lemma}
\begin{proof}
Since $H_b$ commutes with $\tau_{\gamma}^b$, so does $P_b$. We only need to prove this for $\Pi_b$. Using the formula $\Psi_{j,\gamma,b}=\tau_{\gamma}^b (u_{j,b})$, we can express the kernel of $\Pi_b$ (see \eqref{ianuarie10} as: 
$$\Pi_b(\xy,\xy')=\sum_{\gamma'\in\Gamma}\sum_{j=1}^N e^{ib\phi(\xy,\gamma')} u_{j,b}(\xy-\gamma') e^{-ib\phi(\xy',\gamma')} \overline{u_{j,b}(\xy'-\gamma')}. $$
We have:
\begin{equation}\label{febry11}
[\tau_\gamma^b\Pi_b](\xy,\xy')=\sum_{\gamma'\in\Gamma}\sum_{j=1}^N e^{ib\phi(\xy,\gamma)}e^{ib\phi(\xy-\gamma,\gamma')} u_{j,b}(\xy-\gamma'-\gamma) e^{-ib\phi(\xy',\gamma')} \overline{u_{j,b}(\xy'-\gamma')}
\end{equation}
and 
$$[\Pi_b\tau_\gamma^b](\xy,\xy')=\sum_{\gamma'\in\Gamma}\sum_{j=1}^N e^{ib\phi(\xy,\gamma')} u_{j,b}(\xy-\gamma') e^{-ib\phi(\xy'+\gamma,\gamma')} e^{ib\phi(\xy'+\gamma,\gamma)}\overline{u_{j,b}(\xy'-\gamma'+\gamma)}.$$
Changing $\gamma'$ with $\gamma'+\gamma$ in the second formula we obtain:
\begin{equation}\label{febry12}
[\Pi_b\tau_\gamma^b](\xy,\xy')=\sum_{\gamma'\in\Gamma}\sum_{j=1}^N e^{ib\phi(\xy,\gamma'+\gamma)} u_{j,b}(\xy-\gamma'-\gamma) e^{-ib\phi(\xy'+\gamma,\gamma'+\gamma)} e^{ib\phi(\xy'+\gamma,\gamma)}\overline{u_{j,b}(\xy'-\gamma')}.
\end{equation}
Using the bi-linearity and antisymmetry of $\phi$ we see that the magnetic phases in \eqref{febry11} and \eqref{febry12} are the same and the two kernels coincide. 
\end{proof}

\vspace{0.5cm}

Denote by $\widetilde{u}_{j,b}:=U_bu_{j,b}$. Since $\tau_{\gamma}^b$ commutes with $U_b$, we obtain that 
$$\Xi_{j,\gamma,b}=U_b \Psi_{j,\gamma,b}=U_b \tau_{\gamma}^b (u_{j,b})=\tau_{\gamma}^b (\widetilde{u}_{j,b})$$ and:
\begin{equation}\label{ianuar23}
P_b=\sum_{j=1}^N\sum_{\gamma\in \Gamma} |\tau_\gamma^b (\widetilde{u}_{j,b})\rangle \langle \tau_\gamma^b (\widetilde{u}_{j,b})|,\quad \Xi_{j,0,b}=\widetilde{u}_{j,b}.
\end{equation}
In particular, this proves \eqref{M-13}. The last thing we need to prove is \eqref{M-14}, or equivalently: 
\begin{equation}\label{febry13}
\overline{\widetilde{u}_{j,b}}=\widetilde{u}_{j,-b}.
\end{equation}
This last identity would be implied by two others:
\begin{equation}\label{febry20}
\overline{{u}_{j,b}}={u}_{j,-b}\quad {\rm and}\quad \overline{{U}_{b}f}={U}_{-b}\overline{f},\; \forall f\in L^2(\R^2).
\end{equation}
In order to prove $\overline{{u}_{j,b}}={u}_{j,-b}$, we see from \eqref{febry14} and \eqref{M-12} that the only thing we need for this is 
\begin{equation}\label{febry21}
\overline{\mathcal{T}_b(\gamma'';j',j)}=\mathcal{T}_{-b}(\gamma'';j',j).
\end{equation}
From \eqref{febry15} and \eqref{M-12} we obtain that $\overline{\mathcal{M}_b(\gamma'';j',j)}=\mathcal{M}_{-b}(\gamma'';j',j)$, hence $\overline{{M}_{b}\psi}={M}_{-b}\overline{\psi}$ for all $\psi\in l^2(\Gamma)\otimes \C^N$. By functional calculus, the same identity is obeyed by ${M}_{b}^{-1/2}$ which proves \eqref{febry21}, hence  $\overline{{u}_{j,b}}={u}_{j,-b}$. Since $\Psi_{j,\gamma,b}=\tau_{\gamma}^b (u_{j,b})$ we immediately get $\overline{\Psi_{j,\gamma,b}}=\Psi_{j,\gamma,-b}$. From \eqref{ianuarie10} we obtain that $\overline{\Pi_b f}=\Pi_{-b} \overline{f}$ for all $f\in L^2(\R^2)$. Moreover, since $\overline{H_b g}=H_{-b}\; \overline{g}$ for all  $g\in C_0^\infty(\R^2)$, a standard 
argument shows that  $\overline{P_b f}=P_{-b} \overline{f}$ for all $f\in L^2(\R^2)$. Now using the formula \eqref{dech11'} we obtain the second identity in \eqref{febry20} and the proof of Theorem \ref{teorema2} is over.

\section*{Acknowledgements}

 H.C. acknowledges  financial support from Grant 4181-00042 of the Danish Council for Independent Research $|$ Natural Sciences, and from a Bitdefender Invited Professor Scholarship with IMAR, Bucharest. Both I.H. and G.N. acknowledge  financial support from VELUX Visiting Professor Program and from Aalborg University. 
 
The authors thank D. Monaco and G. Panati for several very useful discussions.

\end{document}